\documentclass[journal, 12pt, onecolumn, draftclsnofoot, hidelinks,comsoc,a4paper]{IEEEtran} %
\usepackage[left=2cm,right=2cm,bottom=2.4cm,top=2.4cm]{geometry}
\usepackage{amsmath,amsfonts}
\usepackage{algorithmic}
\usepackage{algorithm}
\usepackage{array}
\usepackage[caption=false,font=normalsize,labelfont=sf,textfont=sf]{subfig}
\usepackage{textcomp}
\usepackage{stfloats}
\usepackage{multirow}
\usepackage{url}
\usepackage{verbatim}
\usepackage{graphicx}
\usepackage{cite}
\usepackage{acronym}
\usepackage[version=4]{mhchem}
\usepackage{siunitx}
\usepackage{tikz-cd}
\usepackage{amsthm}
\usepackage{dsfont}
\usepackage{mathtools}
\usepackage[capitalise]{cleveref}
\DeclareUnicodeCharacter{2212}{\textendash}
\newtheorem{theorem}{Theorem}
\usepackage{placeins}
\usepackage{xltabular}
\usepackage{tabularx}
\usepackage[flushleft]{threeparttable}
\usepackage[version=4]{mhchem}
\usepackage{enumitem}

\expandafter\def\expandafter\normalsize\expandafter{%
    \normalsize%
    \setlength\abovedisplayskip{0pt}%
    \setlength\belowdisplayskip{8pt}%
    \setlength\abovedisplayshortskip{-8pt}%
    \setlength\belowdisplayshortskip{2pt}%
}

\acrodef{TX}{transmitter}
\acrodef{MC}{molecular communication}
\acrodef{TDD}{targeted drug delivery}
\acrodef{ODE}{ordinary differential equation}
\acrodef{SM}{signaling molecule}
\acrodef{FDM}{finite difference method}
\acrodef{RX}{receiver}
\acrodef{PR}{proteorhodopsin}
\acrodef{ATP}{adenosine triphosphate}
\acrodef{ND}{nanodevice}
\acrodef{EM}{electromagnetic}
\acrodef{BSK}{binary shift keying}
\acrodef{IoBNT}{Internet of bionano-things}
\acrodef{LED}{light-emitting diode}
\acrodef{SVS}{single vesicle system}
\acrodef{MVS}{multiple vesicle system}

\newcommand{\Hplus}[0]{\mathrm{H^{+}}}
\newcommand{\Sub}[0]{\mathrm{S}}
\newcommand{\I}[0]{\mathrm{I}}
\renewcommand{\L}[0]{\mathrm{B}}
\newcommand{\NA}[0]{\mathrm{N}_{\mathrm{A}}}

\newcommand{\Chin}[1]{C^{\mathrm{H^{+}}}_{\mathrm{in}}(#1)}
\newcommand{\Chout}[1]{C^{\mathrm{H^{+}}}_{\mathrm{out}}(#1)}

\newcommand{\ChinZERO}[0]{C^{\mathrm{H^{+}}}_{\mathrm{in, 0}}}
\newcommand{\Csin}[1]{C^{\mathrm{S}}_{\mathrm{in}}(#1)}
\newcommand{\CsinZERO}[0]{C^{\mathrm{S}}_{\mathrm{in, 0}}}
\newcommand{\Csout}[1]{C^{\mathrm{S}}_{\mathrm{out}}(#1)}

\newcommand{\Csouttot}[1]{C^{\mathrm{S}}_{\mathrm{out,tot}}(#1)}

\newcommand{\ChoutZERO}[0]{C^{\mathrm{H^{+}}}_{\mathrm{out, 0}}}

\newcommand{\ChoutM}[2]{C^{\mathrm{H^{+}}}_{\mathrm{out},#1}(#2)}
\newcommand{\ChoutTot}[1]{C^{\mathrm{H^{+}}}_{\mathrm{out,tot}}(#1)}
\newcommand{\CsinM}[2]{C^{\mathrm{S}}_{\mathrm{in},#1}(#2)}
\newcommand{\CsoutM}[2]{C^{\mathrm{S}}_{\mathrm{out},#1}(#2)}
\newcommand{\CsoutQ}[2]{C^{\mathrm{S}}_{\mathrm{out,tot},#1}(#2)}

\newcommand{\CZERO}[0]{C_{0}} %

\newcommand{\CiinM}[2]{C^{\mathrm{I}}_{\mathrm{in},#1}(#2)}

\newcommand{\CioutM}[2]{C^{\mathrm{I}}_{\mathrm{out},#1}(#2)}

\newcommand{\ihp}[1]{i^{\mathrm{H^{+}}}_{\mathrm{E}}(#1)}
\newcommand{\ihl}[1]{i^{\mathrm{H^{+}}}_{\mathrm{L}}(#1)}
\newcommand{\ihsym}[1]{i^{\mathrm{H^{+}}}_{\mathrm{R}}(#1)}
\newcommand{\issym}[1]{i^{\mathrm{S}}_{\mathrm{R}}(#1)}

\newcommand{\jhla}[0]{j^{a}_{\mathrm{L}}}
\newcommand{\jhlb}[0]{j^{b}_{\mathrm{L}}}
\newcommand{\jhpa}[0]{j^{a}_{\mathrm{P}}}
\newcommand{\jhpb}[0]{j^{b}_{\mathrm{P}}}
\newcommand{\jhsymb}[0]{j^{b}_{\mathrm{Sym}}(t)}

\newcommand{\ghp}[0]{\gamma_{\mathrm{P}}}
\newcommand{\ghl}[0]{\gamma_{\mathrm{L}}}
\newcommand{\ghsym}[0]{\gamma^{\mathrm{H^{+}}}_{\mathrm{Sym}}}
\newcommand{\gssym}[0]{\gamma^{\mathrm{S}}_{\mathrm{Sym}}}
\newcommand{\gssymT}[1]{\gamma^{\mathrm{S}}_{\mathrm{Sym}}(#1)}
\newcommand{\ghsymT}[1]{\gamma^{\mathrm{H^{+}}}_{\mathrm{Sym}}(#1)}

\newcommand{\ghlM}[1]{\gamma_{\mathrm{L},#1}}
\newcommand{\ghlHatM}[1]{\hat{\gamma}_{\mathrm{L},#1}}

\newcommand{\ghpHat}[0]{\hat{\gamma}_{\mathrm{P}}}
\newcommand{\ghlHat}[0]{\hat{\gamma}_{\mathrm{L}}}
\newcommand{\ghlHatMean}[0]{\bar{\hat{\gamma}}_{\mathrm{L}}}

\newcommand{\gssymHat}[0]{\hat{\gamma}^{\mathrm{S}}_{\mathrm{Sym}}}

\newcommand{\fSym}[0]{\nu_{\mathrm{Sym}}}

\newcommand{\Nh}[0]{N^{\mathrm{H^{+}}}}

\newcommand{\NsM}[1]{N^{\mathrm{S}}_{#1}}

\newcommand{\np}[0]{n_{\mathrm{P}}} %
\newcommand{\npMean}[0]{\bar{n}_{\mathrm{P}}} %
\newcommand{\npM}[1]{n_{\mathrm{P},#1}} %
\newcommand{\ntotM}[1]{n_{\mathrm{tot},#1}} %
\newcommand{\ntot}[0]{n_{\mathrm{tot}}}
\newcommand{\thresh}[0]{\xi}

\newcommand{\Km}[0]{\mathrm{K}_{\mathrm{m}}}
\newcommand{\nsym}[0]{n_{\mathrm{Sym}}}
\newcommand{\nsymMean}[0]{\bar{n}_{\mathrm{Sym}}}
\newcommand{\nsymM}[1]{n_{\mathrm{Sym},#1}}

\newcommand{\din}[0]{d_{\mathrm{in}}}
\newcommand{\dinMean}[0]{\bar{d}_{\mathrm{in}}}
\newcommand{\dinM}[1]{d_{\mathrm{in},#1}}
\newcommand{\dmem}[0]{d_{\mathrm{mem}}}
\newcommand{\Vout}[0]{V_{\mathrm{out}}}
\newcommand{\VoutM}[1]{V_{\mathrm{out},#1}}
\newcommand{\Vin}[0]{V_{\mathrm{in}}}
\newcommand{\VinM}[1]{V_{\mathrm{in},#1}}
\newcommand{\VIN}[0]{\mathcal{V}_{\mathrm{in}}}
\newcommand{\VOUT}[0]{\mathcal{V}_{\mathrm{out}}}
\newcommand{\VOUTtot}[0]{\mathcal{V}_{\mathrm{out,tot}}}
\newcommand{\Vouttot}[0]{V_{\mathrm{out,tot}}}
\newcommand{\VINM}[1]{\mathcal{V}_{\mathrm{in},#1}}
\newcommand{\VOUTM}[1]{\mathcal{V}_{\mathrm{out},#1}}

\newcommand{\step}[0]{\mathrm{\Delta}}

\newcommand{\attSec}[1]{\vartheta_{\mathrm{buf}}(#1)}

\newcommand{\ChinSwitch}[0]{C^{\mathrm{H^{+}}}_{\thresh}}
\newcommand{\ChoutSwitch}[0]{C^{\mathrm{H^{+}}}_{\mathrm{out},\thresh}}
\newcommand{\CiinSwitch}[1]{C^{\I}_{\thresh,m}}
\newcommand{\ChinEq}[0]{C^{\mathrm{H^{+}}}_{\mathrm{in,eq}}}
\newcommand{\ChinStart}[1]{C^{\mathrm{H^{+}}}_{\mathrm{in},0}(#1)}
\newcommand{\CsinStart}[1]{C^{\mathrm{S}}_{\mathrm{in},0}(#1)}

\newcommand{\tji}[2]{t^{(#1)}_{#2}}
\newcommand{\tj}[1]{t^{(#1)}}
\newcommand{\tsp}[1]{\tji{1}{#1}}
\newcommand{\tspNoi}[0]{\tj{1}}

\newcommand{\tssym}[1]{\tji{2}{#1}}
\newcommand{\tssymNoi}[0]{\tj{2}}

\newcommand{\tep}[1]{\tji{3}{#1}}
\newcommand{\tepNoi}[0]{\tj{3}}
\newcommand{\tesym}[1]{\tji{4}{#1}}
\newcommand{\tesymNoi}[0]{\tj{4}}

\newcommand{\tsymX}[2]{t^{(#1)}_{#2}}

\newcommand{\tSecStart}[1]{\tau(#1)}

\newcommand{\lamW}[1]{\mathrm{W}\left \{ #1 \right \}}

\newcommand{\Chplus}[0]{C^{\Hplus}}

\newcommand{\ka}[0]{k_{\mathrm{D}}}
\newcommand{\kPlus}[0]{k_{\mathrm{+}}}
\newcommand{\kMinus}[0]{k_{\mathrm{-}}}

\newcommand{\Ci}[0]{C^{\I}}
\newcommand{\Cl}[0]{C^{\L}}
\newcommand{\Cil}[0]{C^{\I \L}}

\newcommand{\der}[0]{\mathrm{d}}

\newcommand{\iieM}[2]{i^{\mathrm{I}}_{\mathrm{E},#1}(#2)}

\newcommand{\iirM}[2]{i^{\mathrm{I}}_{\mathrm{R},#1}(#2)}
\newcommand{\isr}[1]{i^{\mathrm{S}}_{\mathrm{R}}(#1)}
\newcommand{\isrM}[2]{i^{\mathrm{S}}_{\mathrm{R},#1}(#2)}

\newcommand{\iilM}[2]{i^{\mathrm{I}}_{\mathrm{L},#1}(#2)}

\newcommand{\iieNotM}[1]{i^{\mathrm{I}}_{\mathrm{E},#1}}
\newcommand{\iirNotM}[1]{i^{\mathrm{I}}_{\mathrm{R},#1}}
\newcommand{\isrNotM}[1]{i^{\mathrm{S}}_{\mathrm{R},#1}}
\newcommand{\iilNotM}[1]{i^{\mathrm{I}}_{\mathrm{L},#1}}

\newcommand{\NiM}[1]{N^{\mathrm{I}}_{#1}}

\newcommand{\nves}[0]{n_{\mathrm{ves}}}
\newcommand{\nEx}[0]{n_{\mathrm{exp}}}
\newcommand{\nMod}[0]{n_{\mathrm{mod}}}

\newcommand{\Truncnormal}[1]{\mathcal{TN}(#1)}
\newcommand{\Lognormal}[1]{\mathcal{LN}(#1)}
\newcommand{\Binom}[0]{\mathrm{Binom}}

\newcommand{\probPump}[0]{p_\mathrm{P}}

\makeatletter
\newcommand*{\org@overidelabel}{}
\let\org@overridelabel\@verridelabel
\@ifpackagelater{acronym}{2015/03/21}{%
  \renewcommand*{\@verridelabel}[1]{%
    \@bsphack
    \protected@write\@auxout{}{\string\AC@undonewlabel{#1@cref}}%
    \org@overridelabel{#1}%
    \@esphack
  }%
}{%
  \renewcommand*{\@verridelabel}[1]{%
    \@bsphack
    \protected@write\@auxout{}{\string\undonewlabel{#1@cref}}%
    \org@overridelabel{#1}%
    \@esphack
  }%
}
\makeatother

\makeatletter
\long\def\@makecaption#1#2{\ifx\@captype\@IEEEtablestring%
    \footnotesize\begin{center}{\normalfont\footnotesize #1}\\
        {\normalfont\footnotesize\scshape #2}\end{center}%
    \@IEEEtablecaptionsepspace
    \else
    \@IEEEfigurecaptionsepspace
    \setbox\@tempboxa\hbox{\normalfont\footnotesize {#1.}~~ #2}%
    \ifdim \wd\@tempboxa >\hsize%
    \setbox\@tempboxa\hbox{\normalfont\footnotesize {#1.}~~ }%
    \parbox[t]{\hsize}{\normalfont\footnotesize \noindent\unhbox\@tempboxa#2}%
    \else
    \hbox to\hsize{\normalfont\footnotesize\hfil\box\@tempboxa\hfil}\fi\fi}
\makeatother

\begin{document}

\hyphenation{co-poly-mer-somes}

\bstctlcite{BSTcontrol}

\title{Practical Transmitters for Molecular Communication: Functionalized Nanodevices Employing Cooperative Transmembrane Transport Proteins\vspace*{-0.5cm}}
\author{\IEEEauthorblockN{Teena tom Dieck, Lukas Brand, Lea Erbacher, Daniela Wegner, Sebastian Lotter, Kathrin Castiglione, Robert Schober, and Maximilian Sch\"afer 
}
\vspace*{-1.3cm}}

\IEEEpubid{0000--0000/00\$00.00~\copyright~2021 IEEE}

\maketitle

\nocite{tomDieck2024}

\begin{abstract}
\let\thefootnote\relax\footnotetext{This manuscript has been presented in part at the 2024 ACM International Conference on Nanoscale Computing and Communication \cite{tomDieck2024}.}
This paper introduces a novel optically controllable \ac{MC} \ac{TX} design based on vesicular \acp{ND}. The \acp{ND} are functionalized for the controlled release of \aclp{SM} via transmembrane proteins. The proposed design contributes to overcoming the current barrier between \ac{MC} theory and practical implementation, as all components of the system are chemically realizable. The \acp{ND} possess an optical-to-chemical conversion capability, therefore, the proposed \acp{ND} can be employed as externally controllable \acp{TX} in various \ac{MC} systems. The proposed \ac{ND} design comprises two cooperating modules, namely an energizing module and a release module, and, depending on the specific choices for the modules, allows for the release of different types of \aclp{SM}. After introducing the general system model for the proposed realistic \ac{TX} design, we provide a detailed mathematical analysis of a specific \ac{TX} realization. In particular, we derive both an exact and a closed-form approximate analytical solution for the concentration of the released \aclp{SM} and validate our results by comparison with a numerical solution. Moreover, we model the impact of a buffering medium, which is typically present in liquid environments, e.g., in experimental settings or in in-body applications. This allows the evaluation of the feasibility of our proposed \ac{TX} design in practical chemical implementations. We consider various forms of parameter randomness occurring during vesicle synthesis, i.e., deviations which are unavoidable during experiments. We show that considering random distributions of the parameter values, such as the \ac{ND} size, the number of incorporated proteins on the vesicle surface, and the vesicle membrane permeability, is crucial for an adequate kinetic analysis of the system. Moreover, the proposed analytical and closed-form models for multiple and single \ac{ND} scenarios facilitate system parameter optimization, which can accelerate the realization of the proposed \ac{MC} \ac{TX} in future wet lab experiments.
\end{abstract}

\acresetall

\section{Introduction}

\begin{figure}
    \centering
    \includegraphics[width=0.85\textwidth]{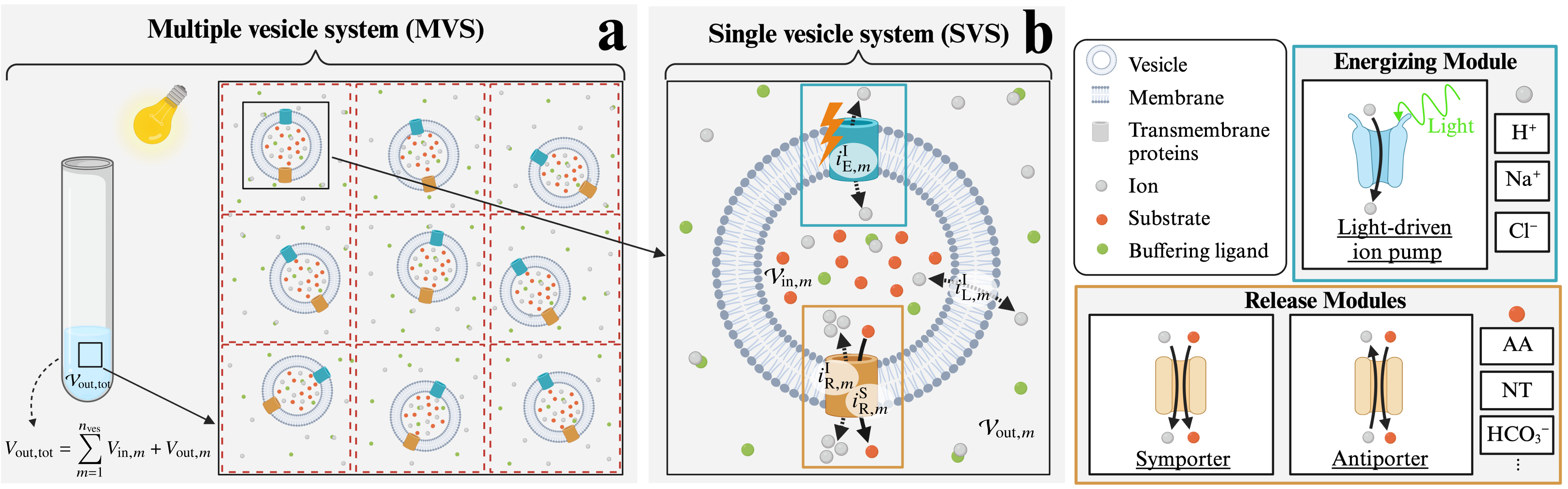}
    \vspace*{-0.4cm}
    \caption{\small General system model for the proposed \acs{ND}-based \acs{TX}. (a): Experimental set up of an \acs{MVS} where $\nves$ \acsp{ND} are contained in volume $\VOUTtot$. (b): Detailed view of \acs{SVS} $m \in \{1,2,\ldots,\nves\}$ in the \acs{MVS}. Variables $\iieNotM{m}$, $\iilNotM{m}$, $\iirNotM{m}$, and $\isrNotM{m}$ denote the flux of ion $\I$ caused by the energizing module, the leakage flux of $\I$, and the flux of $\I$ and substrate $\Sub$ caused by the release module in vesicle $m$, respectively. The possible flux directions between the intravesicular volume, $\VIN$, and the extravesicular volume, $\VOUT$, are indicated by arrows. The complex that the buffering ligand $\L$ may form with $\I$ is also depicted (complex of green and grey dots). The energizing and release module in each vesicle consists of \emph{multiple} transmembrane transport proteins. Examples of suitable proteins and their transported cargo molecules are shown on the right-hand side. Abbreviations: AA = amino acid, NT = neurotransmitter. Created with BioRender.com.}
    \label{fig:sys_model}
\end{figure}

\Ac{MC} is a burgeoning research area in the field of communication engineering and focuses on the development of communication systems that use \acp{SM} as information carriers \cite{Nakano2013,Farsad2016}. Diverging from conventional \ac{EM} wave--based communication, \ac{MC} has emerged as a novel paradigm, with the potential to facilitate communication in scenarios where \ac{EM} wave--based methods face limitations, such as in liquid environments or at nanoscale. Therefore, \ac{MC} offers numerous revolutionary prospective applications including health monitoring, \ac{TDD}, or the detection of toxic agents \cite{Nakano2013,Felicetti2016,Nakano2012}. 
The successful deployment of \ac{MC} systems largely depends on the development of \emph{practically realizable} and \emph{controllable} \ac{TX} designs that ideally support various types of \acp{SM} and are applicable in different environments (e.g., flow- and diffusion-dominated) to enlarge the range of future applications. However, the majority of research on \acp{TX} in \ac{MC} is theoretical and often relies on unrealistic assumptions such as perfect controlability of the \ac{TX} release dynamics or instantaneous release of \acp{SM} \cite{Noel2016,Meng2014}. Recently, more realistic \ac{TX} models have been considered. 
In \cite{Ahmadzadeh2022}, a molecule harvesting \ac{TX}, which is capable of (re-)uptake and release of \acp{SM}, was proposed. In \cite{Schaefer2022}, the controlled release of \acp{SM} by pH-driven membrane permeability switches was considered. The authors of \cite{Huang2022} introduced a modified cell--based membrane fusion \ac{TX} design, which allows for the release of \acp{SM} by the fusion of \ac{SM}-carrying vesicles and the membrane of the \ac{TX}. Additionally, a microscale flow-based \ac{TX} design based on chemical reactions taking place in a microfluidic chip was introduced in \cite{Bi2020}. Similarly, in \cite{Bolhassan2024}, a microfluidic \ac{TX} utilizing hydrodynamic gating was proposed. Nevertheless, there is a lack of externally controllable nanoscale \ac{TX} designs that can also operate in absence of flow. For this scenario, vesicle-based \ac{TX} designs leveraging transmembrane transport proteins are suitable. Initial studies for such designs have been conducted in the \ac{MC} literature: In \cite{Arjmandi2016}, the release of ions from a \ac{ND} via ligand- and voltage-gated ion channels was considered and, in \cite{Grebenstein2019}, an \ac{MC} testbed was presented that utilized bacteria expressing light-driven ion pumps as externally controllable \acp{TX} for the release of protons, showing that \ac{ND}-based optically controllable \acp{TX} are feasible in practice. While the authors of \cite{Arjmandi2016} and \cite{Grebenstein2019} focused on bit transmission as use case for \ac{MC}, where ions are suitable as \acp{SM}, other applications such as \ac{TDD} may require more complex \acp{SM} (e.g., drug molecules). Such \acp{SM} often cannot be transported by externally controllable transport proteins \cite{Soldner2020}. Therefore, to enable the release of these \acp{SM}, the cooperation of different transport proteins is required, as reported in \cite{Soldner2020,Stauffer2021}. Thus, in the proposed \ac{TX} design, one type of protein operates as \textit{energizing module} powering the second type of protein, which serves as \textit{release module} for \acp{SM}. The energizing module facilitates the conversion of external light energy supplied by a \ac{LED} to a chemical concentration gradient between the intra- and extravesicular space using light-driven ion pumps. This gradient then drives the release module which enables the release of \acp{SM} using ion/\ac{SM} co-transporters. Some experimental work has been conducted on the incorporation of optically controllable, i.e., light-driven, transmembrane proteins into synthetic vesicle membranes \cite{Goers2018,Harder2024}, showing that the synthesis of vesicle-based functionalized \acp{ND} is feasible. 

Another shortcoming of existing \ac{MC} \ac{TX} models based on \acp{ND} is that it is assumed that the \ac{TX} consists of a single \ac{ND} only. Experimentally, synthetic and natural vesicle concentrations are typically around $10^{10}\si{\per\milli\liter}$ to $10^{13}\si{\per\milli\liter}$ \cite{Scarpa2016,Johnsen2019} and the extraction, design, and analysis of single vesicles is infeasible due to their small size. Thus, considering \acp{TX} consisting of single \acp{ND} is unrealistic. In this work, we therefore consider \emph{\acfp{MVS}}, which are \ac{TX} systems consisting of many \acp{ND}. Due to the inherent randomness in the production process of vesicles, realistic vesicle-based \acp{ND} exhibit very heterogeneous behavior. It has been found that physical parameters such as the vesicle diameter and the permeability of the vesicle membrane vary among individual \acp{ND} \cite{Poschenrieder2017,Guha2021}. Other parameters that are expected to vary between vesicles, including the number of transmembrane proteins per vesicle, cannot be studied experimentally for single vesicles in a comprehensive and cost-efficient manner, and only estimates for the mean value among many vesicles exist \cite{Gaitzsch2019}. In this work, we take into account this heterogeneous nature of vesicles in \acp{MVS} and evaluate the effect of random parameter values on the variance of \ac{SM} release across different experiments. \acp{MVS} are compared to the simplified \emph{\acf{SVS}} model, where the \ac{TX} consists of exactly one \ac{ND}, as commonly assumed in the literature. 

\noindent Furthermore, in-body fluid systems, e.g., the bloodstream, and most chemical experimental systems rely on buffers for the stabilization of ion concentrations \cite{Ellison1958}. Whilst often disregarded in \ac{MC} models, we also study the influence of this realistic environmental effect on the operation of our proposed \ac{TX} design, i.e., we model the effect a buffer has on the ion concentration.

In this paper, we introduce a realistic externally controllable \ac{TX} design based on multiple vesicular \acp{ND} that are functionalized for the controlled release of \acp{SM} by transmembrane proteins. The design based on two types of cooperative transmembrane transport proteins overcomes the existing drawbacks of previously proposed \acp{TX} by enabling the controlled release of \acp{SM} of various kinds. The combination of cooperating energizing and release modules for increased \ac{SM} versatility, which has not been analyzed in the \ac{MC} literature yet, increases the range of future applications of \ac{ND}-based \acp{TX}. The proposed design might be suitable as a controllable \ac{TX} for \ac{TDD} or information transmission by encoding information in the change of the \ac{SM} concentration over time. 
\\ \noindent The main contributions of this work are as follows:

\begin{itemize}
    \item We investigate the design of modular, externally controllable \acp{TX} capable of releasing different types of \acp{SM} and operating under realistic environmental conditions.
    \item We propose a realistic system model to study the experimental deployment of such \acp{TX} comprising multiple \acp{ND}, where the inherent parameter randomness caused by the \ac{ND} synthesis process is taken into account.
    \item We derive analytical and closed-form expressions for the \ac{SM} release, analyze a possible practical realization of the proposed \ac{TX}, and evaluate the impact of important system parameters on the \ac{SM} release characteristics.
\end{itemize}

\noindent The proposed \ac{TX} design was introduced in the conference version \cite{tomDieck2024} of this paper. Here, we significantly extend \cite{tomDieck2024} by considering realistic \acp{MVS}, deriving mathematical models for the reaction to the used light signal, and analyzing the effect of the inherent randomness of the system parameters known from experimental realizations of vesicle-based \acp{ND}.
\noindent The remainder of this paper is organized as follows. In \Cref{sec:model}, the proposed  \ac{TX} design is introduced in detail and a mathematical description as well as possible biological realizations for the energizing and release modules are provided. In \Cref{sec:study}, an \ac{SVS} is analyzed and corresponding solutions for the \ac{SVS} model are derived in \cref{sec:solutions}. In \cref{sec:extension_mvs}, we extend the analyses to the \ac{MVS} scenario and study the impact of parameter randomness. In \Cref{sec:sim}, we compare the results obtained for different solution approaches. Finally, we draw conclusions in \Cref{sec:conclusion}. Note that \Cref{tab:variables} provides an overview of the used variables and indices for the reader's convenience.

\section{System Model}
\label{sec:model}

In this section, we introduce the system model underlying our analyses. We first provide an overview of the proposed \acp{MVS} and \acp{SVS}. Then, we introduce our modeling assumptions and derive a system of \acp{ODE} modeling the \ac{MVS}. Additionally, we introduce a model for the buffer.

\subsection{Multiple and Single Vesicle Systems}
\label{sec:mvs_and_svs}

\noindent \Cref{fig:sys_model}a shows the proposed \ac{MVS} containing $\nves \in \mathbb{N}_{0}$ \acp{SVS}, i.e., single \acp{ND} featuring both energizing and release modules (right panels). Here, $\mathbb{N}_{0}$ denotes the set of non-negative integers. Note that each \ac{ND} consists of exactly one functionalized vesicle. Therefore, the terms vesicle and \ac{ND} are interchangeable in this work. In \cref{fig:sys_model}, the \ac{SM} (or \emph{substrate} $\Sub$) is depicted in orange. The ions $\I$ creating the gradients required for $\Sub$ release are depicted in gray. As shown in \cref{fig:sys_model}b, an \ac{MVS} of volume $\Vouttot$ containing $\nves$ \acp{ND} can be compartmentalized into multiple \acp{SVS} consisting of the respective extravesicular space $\VOUTM{m}$ of size $\VoutM{m}$ and the respective intravesicular space $\VINM{m}$ of size $\VinM{m}$, where $m \in \{1,2,\ldots,\nves\}$. \ac{ND} $m$ has a spherical shape and a lipid or polymer membrane, which enables the encapsulation of \acp{SM} as cargo in $\VINM{m}$. The vesicle membrane is semi-permeable, i.e., it allows the translocation of some molecules between $\VINM{m}$ and $\VOUTM{m}$. The permeability of the membrane to a specific molecule depends on various factors, including the size of the molecule, with smaller molecules corresponding to a higher membrane permeability. The resulting net flux of ion $\I$ in outward direction over the membrane at time $t$ is denoted by $\iilM{m}{t}$ and is also referred to as \emph{leakage}. The substrate $\Sub$ generally is a larger molecule, e.g., an amino acid, for which the membrane typically has a very low permeability \cite{Chakrabarti1992}. Therefore, we assume the leakage of $\Sub$ over the membrane to be negligible. \Cref{fig:sys_model} also depicts the presence of buffer molecules $\L$ that may form a complex with $\I$ molecules in the aqueous solution.

\noindent The \textit{energizing module} (depicted in blue in \cref{fig:sys_model}) is an energy conversion unit transforming the energy of photons to an electrochemical potential (i.e., a concentration and/or charge gradient\footnote{Note that other sources of energy could also be used to power the ion transport via the energizing module. For instance, \ac{ATP}-coupled transporters utilize the chemical energy stored in the molecule \ac{ATP} as a driving force \cite{Soldner2020}. However, the energy for light-driven ion pumps can be readily supplied externally by an \ac{LED}, and thus, light-driven energizing modules are considered exclusively in this work.}). The energizing module of \ac{ND} $m$, consisting of $\npM{m} \in \mathbb{N}_{0}$ light-driven ion pumping transmembrane proteins, actively transports ions $\I$ over the membrane. The influx caused by the energizing module is denoted by $\iieM{m}{t}$. This flux is generally unidirectional, as the direction of the pumps can be controlled during the insertion process in practice \cite{Goers2018}. 
For the energizing module, several naturally occurring light-driven ion pumps emerge as potential realizations, including proton ($\Hplus$) pumps (such as bacteriorhodopsin \cite{Grebenstein2019} and \ac{PR} \cite{Dioumaev2003}), light-driven chloride ($\ce{Cl-}$) pumps \cite{Schobert1982}, and light-driven sodium ($\ce{Na+}$) pumps \cite{Soldner2020}. 

\noindent The \textit{release module} leverages the ion concentration gradient established by the energizing module as energy supply for the transport of the encapsuled substrate $\Sub$ across the \ac{ND} membrane via $\nsymM{m} \in \mathbb{N}_{0}$ $\I$/$\Sub$ co-transporters. In practice, two main groups of co-transporters exist: \emph{Symporters} transport both molecules $\Sub$ and $\I$ in the same direction (see bottom left box on the right-hand side of \cref{fig:sys_model}), while \emph{antiporters} act as exchangers, i.e., $\Sub$ is transported in the opposite direction as $\I$ (see bottom right box on the right-hand side of \cref{fig:sys_model}). The outfluxes of $\I$ and $\Sub$ caused by the release module are denoted by $\iirM{m}{t}$ and $\isrM{m}{t}$, respectively.
If symporters are used for the release module, an $\I$ gradient from $\VINM{m}$ to $\VOUTM{m}$ has to be established, such that $\Sub$ is transported from $\VINM{m}$ to $\VOUTM{m}$, i.e., released from the vesicle. An $\I$ concentration gradient in the opposite direction is required if antiporters are used. Hence, the required insertion direction of the light-driven $\I$ pumps depends on the type of co-transporter employed.
For some co-transporters, a minimum $\I$ concentration gradient across the membrane is required to facilitate the transport \cite{Nakamura1986}. We denote this threshold concentration by $C^{\I}_{\xi}$, where $\thresh$ is the corresponding threshold. As the specific value and meaning of  $\thresh$ depends on the choice of $\I$, we will define $\thresh$ and $C^{\I}_{\xi}$ more rigorously for the concrete example considered in \cref{sec:study}.
There is a vast number of natural, ion-driven co-transporters capable of transporting complex substrates. Biological examples include $\ce{H+}$/amino acid symporters or $\ce{Na+}$/amino acid symporters \cite{Foltz2005,Ryan2009}, $\ce{Na+}$/neurotransmitter symporters \cite{Soldner2020}, and $\ce{Cl-}$/bicarbonate antiporters \cite{Reithmeier2016}. 
The choice of the specific energizing and release modules depends on the type of ion that is available, as both the light-driven pumps and co-transporters need to be able to transport it. For example, if the desired substrate is an amino acid, $\Hplus$/amino acid symporters can be used in combination with $\Hplus$ pumps to establish the necessary $\Hplus$ gradient.
Although there exist a number of possible combinations of energizing and release modules, the practical insertion of some transport proteins into the vesicle membrane may be challenging and many co-transporters lack a formal kinetic characterization. Thus, for the system analysis and the simulations, we will consider proteins that have already been successfully inserted into synthetic vesicle membranes and/or for which the transport kinetics are known.
\subsection{Modeling Assumptions}
\label{sec:ass}
\noindent For the sake of mathematical tractability, we make the following assumptions. These assumptions are used to set up the system model in the following sections.

\begin{itemize}
    \item[(A1)] We assume that in each \ac{SVS}, the solution is well-stirred (e.g., in a stirred bioreactor \cite{Schirmer2021}) and the total number of ions, $\NiM{m}$, and substrate molecules, $\NsM{m}$, are known. We assume that the concentrations of $\I$ and $\Sub$ are uniform in both $\VOUTM{m}$ (e.g., due to stirring) and $\VINM{m}$, as the diffusion of $\I$ and $\Sub$ is fast in comparison to their transport over the membrane: The expected range of diffusion coefficients of $\I$ and $\Sub$ is around $10^{-12} \si{\square\meter\per\second}$ to $10^{-9} \si{\square\meter\per\second}$ \cite{Ma2005,Rice1985,Hills2011}, which are large in comparison to typical vesicle diameters, which are around $20 \si{\nano\meter}$ to $1000 \si{\nano\meter}$ \cite{Rideau2018}, and transport rates of the membrane proteins, which are around $10^{-3} \si{\per\s}$ to $10^{3} \si{\per\second}$ \cite{Tubbe1992,Lanyi1992}. 
    \item[(A2)] We assume that the light signal emitted by the \ac{LED}, $l(t)$, is binary, i.e., $l(t) \in \{0,1\}$ for all times $t$. Here, $l(t) =1$ indicates that the \ac{LED} is switched on, and $l(t) = 0$ means that the \ac{LED} is switched off.
    \item[(A3)] We assume that the buffer molecule $\L$ only interacts with $\I$, as it has a low affinity to other molecules. This assumption was verified experimentally for some standard buffers \cite{Ferreira2015}.
    \item[(A4)] We assume that the intra- and extravesicular volumes $\VINM{m}$ and $\VOUTM{m}$ of vesicle $m$ only interact with each other and do not influence and are not influenced by any other $\VINM{q}$ or $\VOUTM{q}$ if $q \neq m$ for $q,m \in \{1,2,\ldots,\nves\}$, i.e., we assume \ac{SVS} independence. This assumption is justified for $\VoutM{m} \gg \VinM{m}$, as the fluxes from and into the vesicle only generate very small changes in the extravesicular concentrations in this case.
    \item[(A5)] For the analytical expressions, we compartmentalize the total extravesicular volume $\VOUTM{m}$ uniformly to obtain the extravesicular volumes of \ac{SVS} $m$, i.e., $\VoutM{m} = \Vouttot / \nves$. This assumption is justified as $\VOUTtot$ is much larger than the sum of the size of the intravesicular volumes, i.e., $\Vouttot \gg \sum_{m = 1}^{\nves} \VinM{m}$, and the entire system is well stirred (see (A1)). In our baseline numerical simulations (see \cref{sec:sim}), we implement only one large extravesicular volume to validate this assumption.
\end{itemize}

\subsection{System of \acsp{ODE} Modeling the \acs{MVS}}
\label{sec:odes}

Using assumptions (A1)--(A4), the fluxes of $\I$ and $\Sub$ caused by the energizing module (comprising $\npM{m}$ pumps), release module (comprising $\nsymM{m}$ symporters), and leakage can be used to set up a system of coupled \acp{ODE} describing the system kinetics for the proposed \ac{TX} design for an \ac{SVS} in the \ac{MVS} as follows

\begin{align}
    \VinM{m} \frac{\der \CiinM{m}{t}}{\der t} &= \iieM{m}{t} - \iilM{m}{t} - \iirM{m}{t, \CsinM{m}{t}}, \label{eq:chin} \\ \VinM{m} \frac{\der \CsinM{m}{t}}{\der t} &= - \isrM{m}{t, \CiinM{m}{t}},
    \label{eq:csin}
\end{align}

\noindent for all vesicles $m \in \{1, 2, \ldots , \nves \}$. Note that \eqref{eq:chin} and \eqref{eq:csin} are not coupled for different $m$, as we assume \ac{SVS} independence (see (A4)), but they are coupled for each $m$. Here, $\CiinM{m}{t}$, $\CsinM{m}{t}$, $\iieM{m}{t}$, $\iirM{m}{t, \CsinM{m}{t}}$, and $\isrM{m}{t, \CiinM{m}{t}}$ are the intravesicular concentrations of $\I$ and $\Sub$ in \si{\mol\per\cubic\meter}, the influx of $\I$ caused by the energizing module, the outflux of $\I$ caused by the release module, and the outflux of $\Sub$ caused by the release module of \ac{ND} $m$, respectively. All fluxes are measured in \si{\mol\per\s}. If the release module consists of antiporters, the signs of the fluxes in \eqref{eq:chin} need to be switched. Note that the coupled system of \acp{ODE} in \eqref{eq:chin} and \eqref{eq:csin} only considers the concentrations in $\VINM{m}$, but is sufficient to characterize the entire \acp{SVS}, as $\NsM{m}$ and $\NiM{m}$ are constant and known (see (A4)). The concentrations in $\VOUTM{m}$ can thus be derived from those in $\VINM{m}$, i.e., the $\I$ concentration $\CioutM{m}{t} = (\NiM{m} - \CiinM{m}{t}\VinM{m})/\VoutM{m}$ and the $\Sub$ concentration $\CsoutM{m}{t} = (\NsM{m} - \CsinM{m}{t}\VinM{m})/\VoutM{m}$. According to \eqref{eq:chin} and \eqref{eq:csin}, all \acp{SVS} can be analyzed \textit{individually} before determining their joint impact on the extravesicular volume $\VOUTtot = \bigcup_{m=1}^{\nves} \VOUTM{m}$ by averaging the impact of the individual \acp{SVS}. 
To this end, the $\Sub$ concentration in $\VOUTtot$, can be obtained from the individual \acp{SVS} in subvolumes $\VOUTM{m}$ as follows

\begin{equation}
    \Csouttot{t} = \frac{1}{\Vouttot} \sum_{m=1}^{\nves} \VoutM{m} \CsoutM{m}{t} \stackrel{\text{(A5)}}{=} \frac{1}{\nves} \sum_{m=1}^{\nves} \CsoutM{m}{t}.
    \label{eq:csouttot}
\end{equation} 

\subsection{Buffering Effects}
\label{sec:model_buff}

\noindent The system of \acp{ODE} in \eqref{eq:chin} and \eqref{eq:csin} does not consider any buffering effects, yet. However, as metal ion or pH buffers are used in most experimental environments and are present, e.g., in in-body fluids \cite{Ellison1958}, their effect has to be taken into account.  We consider the following reversible reaction between $\I$ and a buffering ligand $\L$

\begin{align}
	\ce{IB <=>[$k_-$][$k_{+}$] I + B},
	\label{eq:eq}
\end{align}

\noindent where $\I\L$, $\kMinus$, and $\kPlus$ are the complex formed by the ion and the ligand, the unbinding rate constant of $\I$ from $\L$, and the binding rate constant of $\I$ to $\L$, respectively. The dissociation constant of the ligand in equilibrium is $\ka = \kMinus / \kPlus$ in \si{\mol\per\cubic\meter}. Using the mass action law \cite{Po2001}, we obtain the concentration of $\I$ in the buffered system as follows
\begin{equation}
    -\log (\Ci) =  - \log (\ka) + \log  (  \Cl) - \log (\Cil),
    \label{eq:hh-eq}
\end{equation}
\noindent where $C^{\mathrm{X}}$ denotes the concentration of molecule $\mathrm{X}$. Generally, the buffer molarity is given by $\CZERO = \Cil + \Cl$ and remains constant. Thus, in equilibrium and for a given $\I$ concentration, $\Ci$ and $\Cil$ can be deduced from \eqref{eq:hh-eq}. Note that for $\I \;\hat{=}\; \Hplus$ ligand $\L$ would be a base and $-\log (\Chplus) = \mathrm{pH}$. In this case, \eqref{eq:hh-eq} specializes to the well-known Henderson-Hasselbalch equation \cite{Ellison1958}. We do not explicitly model the concentration of the buffer molecules in our system of \acp{ODE}, as the required extension of \eqref{eq:chin} and \eqref{eq:csin} to account for the concentration of four types of molecules becomes intractable. Instead, \eqref{eq:hh-eq} will be incorporated into the numerical simulations of \eqref{eq:chin} and \eqref{eq:csin} and serve as a ground truth for the impact of the buffer. However, in our proposed analytical models, the impact of the buffer is approximated (see \Cref{sec:buffer}).

\section{Operation and Analysis of Single Vesicle System}
\label{sec:study}

This section investigates a specific realization of the proposed \acp{ND} using light-driven $\Hplus$ pumps (such as \ac{PR} \cite{Dioumaev2003}) and $\Hplus$ symporters (e.g., PAT1 \cite{Foltz2005}) as energizing and release modules, respectively. Hence, $\I \;\hat{=}\; \Hplus$ in the rest of this paper. As the use of light-driven $\Hplus$ pumps allows for a variety of possible release modules and, hence, different $\Sub$ (e.g., amino acids or neurotransmitters), we continue to treat $\Sub$ as a generic substrate molecule. For now, we concentrate on the \ac{SVS} scenario. We begin by discussing the envisioned functionality of a \ac{SVS} in \cref{sec:func}. In \cref{sec:cycle_types}, we discuss the different cycle types that can be observed upon illumination. Then, we propose expressions for the $\Hplus$ and $\Sub$ fluxes in \cref{sec:fluxes}.

\renewcommand\arraystretch{0.9}
\footnotesize
\begin{xltabular}{\textwidth}{|l|X|} 
\caption{\centering{Overview of the most important variables and indices.}}\label{tab:variables} \\
    \hline
    \textbf{Index} & \textbf{Meaning} \\
    \hline\hline
    $\I$ & Ion transported by light-driven ion pumps and co-transporters. For the specific realization considered in \cref{sec:study}, $\I \;\hat{=}\; \Hplus$. \\ \hline
    $\Sub$ & Substrate transported by co-transporters. \\ \hline
    $\mathrm{E}$ & Energizing module \\ \hline
    $\mathrm{R}$ & Release module \\ \hline
    $\mathrm{L}$ & Leakage \\ \hline
    $\mathrm{P}$ & Light-driven ion pumps \\ \hline
    $\mathrm{Sym}$ & Symporters (or antiporters) \\ \hline\hline
    \textbf{Variable} & \textbf{Description} \\
    \hline\hline
    $\mathrm{C}^{\mathrm{X}}_{\mathrm{x},m}(t)$ & Concentration of molecule $\mathrm{X} \in \{\I, \Sub\}$ in the intra- or extravesicular volume ($\mathrm{x} \in \{\mathrm{in}, \mathrm{out}\}$) of vesicle $m$ over time $t$ in \si{\mol\per\cubic\meter}. \\ \hline
    $\mathrm{C}^{\mathrm{X}}_{\mathrm{x},0}$ & Initial concentration of molecule $\mathrm{X}$ in the intra- or extravesicular volume ($\mathrm{x} \in \{\mathrm{in}, \mathrm{out}\}$) of each vesicle. \\ \hline
    $\CZERO$ & Buffer molarity in \si{\mol\per\cubic\meter}. \\ \hline
    $\dinM{m}$ & Inner vesicle diameter of vesicle $m$ in \si{\m}. \\ \hline
    $\dmem$ & Diameter of the vesicle membrane in \si{\m}. \\ \hline
    $\fSym$ & Ratio of $\I$ and $\Sub$ molecules transported by a co-transporter. \\ \hline
    $\ghlM{m},\ghlHatM{m}$ & Effective rate of $\I$ leakage over the membrane and permeability coefficient of the membrane to $\I$ of vesicle $m$ in \si{\cubic\meter\per\s} and \si{\meter\per\s}, respectively. \\ \hline
    $\gamma^{\I}_{\mathrm{P},m}, \ghpHat$ & Effective $\I$ transport rate of the energizing module in vesicle $m$ and effective $\I$ transport rate of one light-driven $\I$ pump in \si{\mol\per\s} and \si{\per\s}, respectively. \\ \hline
    $\gamma^{\mathrm{X}}_{\mathrm{Sym},m}, \hat{\gamma}^{\mathrm{X}}_{\mathrm{Sym}}$ & Effective $\mathrm{X} \in \{\I, \Sub\}$ transport rate of the release module in vesicle $m$ and effective $\mathrm{X}$ transport rate of one co-transport protein in \si{\mol\per\s} and \si{\per\s}, respectively. \\ \hline
    $\ka$ & Buffer dissociation rate in \si{\mol\per\cubic\meter}. \\ \hline
    $\Km$ & Michaelis-Menten constant of the co-transporters in \si{\mol\per\cubic\meter}. \\ \hline
    $l(t)$ & Light signal with $l(t) \in \{0,1\}$ over time $t$. \\ \hline
    $i^{\I}_{\mathrm{Y},m}(t)$ & Flux of $\I$ caused by energizing or release module or leakage ($\mathrm{Y} \in \{\mathrm{E}, \mathrm{R}, \mathrm{L}\}$, respectively) in vesicle $m$ over time $t$ in \si{\mol\per\cubic\meter\per\s}. \\ \hline
    $i^{\Sub}_{\mathrm{R},m}(t)$ & Flux of $\Sub$ caused by the release module in subsystem $m$ over time $t$ in \si{\mol\per\cubic\meter\per\s}. \\ \hline
    $\xi$ & Logarithmic threshold of symport start. \\ \hline
    $\CiinSwitch{m}$ & Symport threshold concentration, i.e., for symporter activity in \si{\mol\per\cubic\meter}. \\ \hline
    $n_{\mathrm{x},m}$ & Number of pumps ($\mathrm{x} = \mathrm{P}$) or co-transporters ($\mathrm{x} = \mathrm{Sym}$) in membrane of vesicle $m$. \\ \hline
    $\ntotM{m}$ & Total number of functional proteins in membrane of vesicle $m$. \\ \hline
    $\nEx$ & Number of simulations. \\ \hline
    $\nves$ & Number of vesicles/\acp{SVS} in an \ac{MVS}. \\ \hline
    $\nMod$ & Number of simulated vesicles/\acp{SVS} in an \ac{MVS} ($\nMod \le \nves$). \\ \hline
    $\tSecStart{t}$ & Start time of the current signal phase at time $t$ in \si{\s}.\\ \hline
    $\tji{\mathrm{x}}{i}$ & Time instants of pumping start ($\mathrm{x} = 1$) and end ($\mathrm{x} = 3$) in cycle $i$ in \si{\s}. \\ \hline
    $\tji{\mathrm{x}}{i,m}$ & Time instants of symport start ($\mathrm{x} = 2$) and end ($\mathrm{x} = 4$) in cycle $i$ in subsystem $m$ in \si{\s}. \\ \hline
    $\attSec{t}$ & Flux attenuation factor caused by the buffer. \\ \hline
    $\mathcal{V}_{\mathrm{x},m},V_{\mathrm{x},m}$ & Intra- ($\mathrm{x} = \mathrm{in}$) and extravesicular ($\mathrm{x} = \mathrm{out}$) spaces and their sizes in \si{\cubic\meter}, respectively. \\ \hline 
    $\VOUTtot, \Vouttot$ & Total extravesicular space and its size in \si{\cubic\meter}, respectively.\\ \hline
\end{xltabular}
\normalsize

\begin{figure*}[thb]
    \centering
    \includegraphics[width=\textwidth]{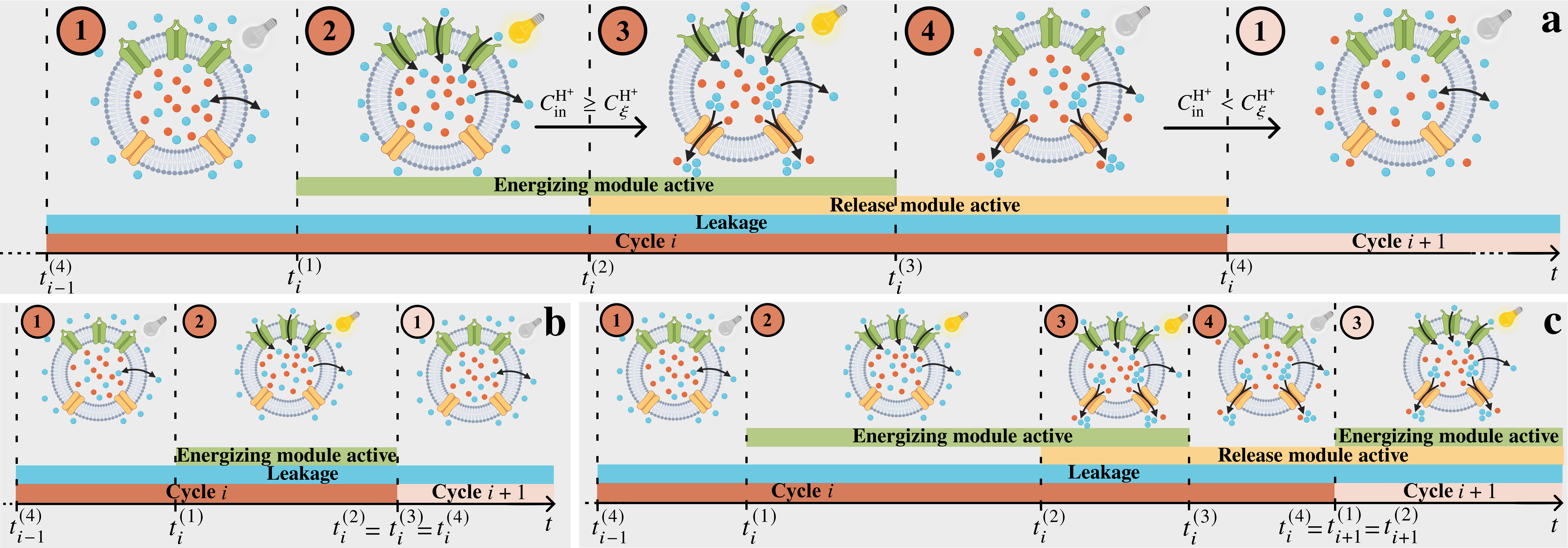}
    \caption{(a): A regular illumination cycle for an \ac{ND} with $\smash{\np = 3}$ and $\smash{\nsym = 2}$ comprising four different cycle phases. The times $\smash{\tsp{i}}$, $\smash{\tssym{i}}$, $\smash{\tep{i}}$, and $\smash{\tesym{i}}$ mark the transitions between two cycle phases. Here, $\smash{\ChinSwitch}$ indicates the symport threshold concentration. (b): A cycle without symporter activity. (c): Two subsequent cycles in between which the symport does not end. Parts of the image were created with BioRender.com.}
    \label{fig:cycle1}
\end{figure*}

\subsection{Operation of Single Vesicle System}
\label{sec:func}

In the following, while considering one \ac{SVS}, we will drop vesicle index $m$ to simplify the notation. To discuss the functionality of the proposed \ac{ND} in detail for the \ac{SVS}, it is helpful to examine its behavior upon different external and internal stimuli. Hence, we consider the different states of the \ac{ND} possible during one illumination cycle (shown in \cref{fig:cycle1}). A cycle describes a time period during which the light source is turned on and off exactly once. A cycle consists of different phases, which are defined as the time periods during which a certain combination of system components (energizing module, release module, and leakage) are active. As the external light signal $l(t) \in \{0,1\}$ can be chosen arbitrarily and usually consists of multiple illumination cycles, variable $i \in \mathbb{N}_{0}$ is used to index the cycles. Variables $\tji{j}{i}$ for $j \in \{1, 2, 3, 4\}$ denote the end times of phase $j$ in cycle $i$, as shown on the axis in \cref{fig:cycle1}a. First, we consider a regular cycle $i$ consisting of the following four phases (P1)--(P4) (see \cref{fig:cycle1}a):

\begin{enumerate}
    \item[(P1)] \textbf{Leakage:} During the first cycle phase the \ac{ND} is not illuminated and both types of transport proteins are inactive, i.e., only the leakage influences the $\Hplus$ flux ($|\ihl{t}| \ge 0$, $\ihp{t} = 0$, $\ihsym{t} = 0$). 
    \item[(P2)] \textbf{Energizing module and leakage:} At $\tsp{i}$ the illumination begins and the proton pumps start transporting $\Hplus$ ($\ihp{t} > 0$). Simultaneously, the increasing $\Hplus$ concentration in $\VIN$ leads to an increasing pH difference between $\VOUT$ and $\VIN$ yielding a leakage outflux of $\Hplus$ ($\ihl{t} > 0$).
    \item[(P3)] \textbf{Energizing and release modules, and leakage:} When the threshold concentration for symporter activity within the vesicle, $\ChinSwitch$, is reached at time $\tssym{i}$, the symporters in the release module become active. The symporters cause an additional outflux of $\Hplus$ and an outward transport of $\Sub$ ($\ihsym{t} > 0, \issym{t} > 0, \ihp{t} > 0, \ihl{t} > 0$).
    \item[(P4)] \textbf{Release module and leakage:} When the illumination ends at time $\tep{i}$ but the intravesicular $\Hplus$ concentration is still above the threshold, i.e., $\Chin{t} > \ChinSwitch$, the symporters remain active while the pumps stop transporting $\Hplus$ ($\ihp{t} = 0$). During this cycle phase, both the symporters and the leakage cause $\Hplus$ outflux and $\Sub$ is transported outwards ($\ihsym{t} > 0, \issym{t} > 0, \ihl{t} > 0$). When $\Chin{t}$ falls below threshold $\ChinSwitch$ at time $\tesym{i}$, cycle $i$ ends and the next cycle $i + 1$ starts. Note that, at the latest, a new cycle begins when the next illumination phase starts (even if the symport has not ended at that point).
\end{enumerate}

\noindent Note that we assume $\tesym{0} = 0$, i.e., the first cycle ($i = 1$) starts at $t = 0$. While $\tsp{i}$ and $\tep{i}$ depend on the external light signal $l(t)$, the symport start and end times, $\tssym{i}$ and $\tesym{i}$, depend exclusively on the $\Hplus$ concentrations and, thus, have to be calculated from $\Chin{t}$, as shown in \Cref{sec:limits}. The cycle index is required for the time variables limiting the cycle phases, but can be omitted for the fluxes and concentrations, which are defined for absolute time $t$. This is possible because the time variables for new cycles are monotonically increasing, i.e., $\tji{j'}{i} \ge \tji{j}{i}$ where $j' > j$, and $\tji{j}{i'} \ge \tji{j}{i}$ where $i' > i$ (see \cref{fig:cycle1}). Generally, different types of illumination cycles can occur and the previously discussed sequence of phases holds only for regular cycles, i.e., cycles of type (a) in \cref{fig:cycle1}. For other cycle types, not all individual phases (P1)--(P4) do necessarily occur. These cycle types are examined in the following.

\subsection{Other Types of Cycles}
\label{sec:cycle_types}

\noindent As mentioned previously, other types of cycles than type (a) can occur. Cycles of type (b) (\cref{fig:cycle1}b), for instance, occur when the illumination period in cycle $i$ is too short to reach $\ChinSwitch$. As the resulting symport duration, $\tesym{i} - \tssym{i}$, is $\SI{0}{\s}$, $\tssym{i} = \tep{i} = \tesym{i}$ follows to ensure that the symport duration is correctly obtained and that all cycle limit times are properly defined and monotonically increasing. 

\noindent Cycle type (c) occurs when the time between two illumination phases is too short for the symport of the first cycle to stop (see \cref{fig:cycle1}c). Hence, the time between the symport of two cycles, $\tssym{i + 1} - \tesym{i}$, is set to $\SI{0}{\s}$. This holds for $\tesym{i} = \tsp{i + 1} = \tssym{i + 1}$. Hereby, it is ensured that all cycle phase limit times are properly defined.

\subsection{Proton and Substrate Fluxes}
\label{sec:fluxes}
\noindent Before deriving solutions for \eqref{eq:chin} and \eqref{eq:csin} in \cref{sec:exact,sec:closed}, mathematical models for the fluxes of $\Hplus$ and $\Sub$ are required. 

\paragraph{Energizing Module}

\noindent We assume that the light-driven proton pumps always operate at maximum effective rate as long as enough $\Hplus$ is available in $\VOUT$ during illumination. As the transport process of $\Hplus$ by light-driven pumps is rate-limited by one reaction, this assumption is well-justified \cite{Bamann2014,Arjmandi2016}. Consequently, the $\Hplus$ flux in \si{\mol\per\s} caused by the energizing module comprising $\np$ proton pumps, $\ihp{t}$, is defined as follows

\begin{equation}
    \ihp{t} = \frac{\Chout{t}}{\ChoutZERO} \ghp \mathds{1}_{\{1\}}(l(t)),
    \label{eq:flux_p}
\end{equation}

\noindent where $\ChoutZERO$, $\ghp$, and $\mathds{1}_{\mathcal{X}}(x)$ are the initial $\Hplus$ concentration in $\VOUT$, the effective rate constant of $\Hplus$ for $\np$ proton pumps, and the indicator function, i.e., $\mathds{1}_{\mathcal{X}}(x) = 1$ if $x \in \mathcal{X}$, and $\mathds{1}_{\mathcal{X}}(x) = 0$ otherwise, respectively. The effective pumping rate of one vesicle, $\ghp = \ghpHat \np/\NA$ in \si{\mol\per\s}, depends on the effective pumping rate of one proton pump $\ghpHat$ in $\si{\per \second}$, the number of pumps $\np$, and the Avogadro constant $\NA = \SI{6.022e23}{\per\mol}$.

\paragraph{Release Module}

\noindent We assume that the symporters are only active if the intravesicular $\Hplus$ concentration, $\Chin{t}$, exceeds threshold $\ChinSwitch$, based on the co-transporter kinetics described in the literature \cite{Nakamura1986,Foltz2005}. The associated $\Sub$ and $\Hplus$ fluxes caused by the release module comprising $\nsym$ symporters can be described as follows

\begin{equation}    
    \issym{t} = \gssym \left ( \frac{\Csin{t}}{\Csin{t} + \Km} \right) \mathds{1}_{[\ChinSwitch,\infty)}(\Chin{t}), \; \; \ihsym{t} =\fSym \issym{t},%
    \label{eq:ihsym}
\end{equation}

\noindent where $\gssym$, $\Km$, and $\fSym \in \mathbb{Q}$ are the effective $\Sub$ symport rate constant of one vesicle, the Michaelis-Menten constant in \si{\mol\per\cubic\meter}, and the ratio of $\Hplus$ to $\Sub$ molecules that are co-transported by a symporter, respectively. Here, $\mathbb{Q}$ denotes the set of rational numbers. As $\fSym$ is fixed and depends on the molecular structure of the co-transporter, the $\Hplus$ and $\Sub$ fluxes caused by the symporters can be deduced from one another by multiplication or division by $\fSym$. The effective symport rate of the release module in one \ac{ND} is given by $\gssym = \gssymHat \nsym / \NA$ in $\si{\mol \per \second}$, where $\gssymHat$ is the effective $\Sub$ transport rate constant of one symporter in \si{\per\s}. In case of $\I = \Hplus$, $\thresh$ corresponds to the pH threshold between $\VIN$ and $\VOUT$ required for the start of the symport \cite{Nakamura1986}. Thus, the threshold concentration required for the symport to start is $\ChinSwitch = \Nh/(\Vout 10^{-\thresh} + \Vin)$, where $\Nh$ is the total number of $\Hplus$ molecules in the \ac{SVS}, and is obtained from the initial system pH and $\thresh$. The expression for $\ChinSwitch$ is obtained by solving the equation $\thresh = \log (\ChinSwitch) - \log (\ChoutSwitch)$, where $\ChoutSwitch = (\Nh - \Vin \ChinSwitch)/\Vout$ is the extravesicular $\Hplus$ concentration when the symport starts. Note that the symporters are naturally only active as long as $\Sub$ is available, i.e., $\Csin{t} > 0$. The symporters become inactive when the vesicle has released all of its cargo. This state is referred to as substrate depletion. 

\paragraph{Leakage}

\noindent The $\Hplus$ flux caused by leakage of $\Hplus$ over the vesicle membrane is given as follows

\begin{equation}        
    \ihl{t} = \ghl \left ( \Chin{t} - \Chout{t} \right ),
    \label{eq:flux_l}
\end{equation}

\noindent where $\ghl = \ghlHat A_{\mathrm{ves}}$ is the membrane permeability with respect to $\Hplus$ in $\si{\cubic \meter \per \second}$. Here, $\ghlHat$ is the permeability coefficient of $\Hplus$ for the vesicle membrane in $\si{\meter \per \second}$ and $A_{\mathrm{ves}}$ is the outer surface area of the vesicle in $\si{\square \meter}$. Note that the leakage flux scales linearly with the $\Hplus$ concentration gradient between $\VIN$ and $\VOUT$.

\section{Analytical Solutions}
\label{sec:solutions}

In this section, we derive analytical solutions for the concentrations of $\Hplus$ and $\Sub$. First, an exact analytical solution is proposed. Then, we derive a closed-form approximation. Subsequently, we provide expressions for the calculation of the cycle-phase limits. Finally, we show how the buffering effect can be incorporated into the analytical and close-form models.

\subsection{Exact Analytical Solution}
\label{sec:exact}
For the analytical solution, \acp{ODE} \eqref{eq:chin} and \eqref{eq:csin} are considered separately for each cycle phase shown in \cref{fig:cycle1}. 
Therefore, we introduce variable $\tSecStart{t}$ to indicate the start time of the current cycle phase

\begin{equation}
    \tSecStart{t} = \max_{i,j} \big\{ \tji{j}{i} \left| \right. \tji{j}{i} < t \big\}, \quad \forall j \in \{1,2,3,4\}.
    \label{eq:tsecstart}
\end{equation}

\noindent Note that, as we define all cycle limits for absolute time, maximization over the cycles themselves is required to ensure that the current/most recent cycle is considered in \eqref{eq:tsecstart}. Moreover, $\Csin{t}$ and $\Chin{t}$ can be expressed in a compact manner, using Proposition 1 (see below) and the cycle phase--dependent variables $a$ and $b$

\begin{equation}
    a = \jhla + \jhpa \mathds{1}_{\{\tsp{i}, \tssym{i}\}}(\tSecStart{t}), \quad
    b = \jhlb + \jhpb \mathds{1}_{\{\tsp{i}, \tssym{i}\}}(\tSecStart{t}),
\end{equation}

\noindent with auxiliary variables $\jhla = \ghl (\Vin^{-1} + \Vout^{-1})$, $\jhpa = \ghp / (\Vout \ChoutZERO)$, $\jhlb = - \ghl \Nh /(\Vout\Vin)$, and $\jhpb = \jhpa \Nh /\Vin$. Here, $a$ and $b$ capture the influence of the energizing module and the leakage on the intravesicular $\Hplus$ concentration. In the following, we use auxiliary function

\begin{equation}
    f(t) := \CsinStart{t} / \Km \, \mathrm{exp}([\CsinStart{t} - \gssym / \Vin (t - \tSecStart{t})]/\Km),
\end{equation}

\noindent where $\CsinStart{t} = \Csin{\tSecStart{t}}$ denotes the initial intravesicular concentration of $\Sub$ in the current cycle phase\footnote{Note that $\CsinZERO := \Csin{0}$ denotes the initial intravesicular concentration of $\Sub$ in the system at $t = 0$.}.

\begin{proof}[Proposition 1] The intravesicular concentration of $\Sub$, $\Csin{t}$, in \eqref{eq:csin} is obtained as follows for each phase $j$ in cycle $i$ \\
\begin{equation}
    \Csin{t} = 
    \begin{cases}
        \Km \lamW{f(t)}, & \text{if } \tssym{i} < t \le \tesym{i},\\
        \CsinStart{t}, & \text{if } \tesym{i - 1} < t \le \tsp{i},
    \end{cases} %
    \label{eq:csin_aci}
\end{equation}
\noindent where $\lamW{\cdot}$ denotes the Lambert W--function, defined by ${\lamW{x}\mathrm{exp}(\lamW{x}) = x}$.
\noindent The intravesicular $\Hplus$ concentration, $\Chin{t}$, in \eqref{eq:chin} is obtained as follows for each phase $j$ in cycle $i$
\begin{equation}
    \Chin{t} = \left[\ChinStart{t} - \alpha(\tSecStart{t}) + \alpha(t)\right] \mathrm{e}^{-a(t - \tSecStart{t})},
    \label{eq:chin_aci}
\end{equation}

\noindent where $\ChinStart{t} = \Chin{\tSecStart{t}}$ is the initial intravesicular $\Hplus$ concentration of current cycle phase $j$ and
\begin{equation}
    \alpha(t) = 
    \begin{cases}
        \int \limits_{\tSecStart{t}}^{t}\left( b - \frac{\ghsym}{\Vin}\frac{\lamW{f(\omega)}}{\lamW{f(\omega)} + 1} \right )\mathrm{e}^{a(\omega - \tSecStart{\omega})} \mathrm{d} \omega, & \text{if }  \tssym{i} < t \le \tesym{i}, \\
        \frac{b}{a} \mathrm{e}^{a(t - \tSecStart{t})}, & \text{if } \tesym{i - 1} < t \le \tsp{i}.
    \end{cases} %
    \label{eq:variedc}
\end{equation}
\renewcommand{\qedsymbol}{}
\end{proof}

\begin{proof}
    Due to space limitations, we provide only a sketch of the proof for Proposition 1. We obtain \eqref{eq:csin_aci} by inserting \eqref{eq:ihsym} into \eqref{eq:csin} and solving for $\Csin{t}$. For a detailed derivation of the solution of the integral of a Michaelis-Menten equation, we refer the reader to \cite[Sec.~2]{Goudar1999}.
    We obtain \eqref{eq:chin_aci} by inserting \eqref{eq:csin_aci} and \eqref{eq:flux_p}--\eqref{eq:flux_l} into \eqref{eq:chin}, resulting in
     \begin{equation}
        \frac{\mathrm{d} \Chin{t}}{\mathrm{d} t} = - a \Chin{t} + b - \mathds{1}_{[\ChinSwitch,\infty)}(\Chin{t}) \frac{\ghsym}{\Vin} \frac{\lamW{f(t)}}{\lamW{f(t)} + 1}.
        \label{eq:chin_aci_d}
    \end{equation}
    \noindent Solving this inhomogeneous \ac{ODE} by variation of the constant yields the time-variant integration constant $\alpha(t)$ in \eqref{eq:variedc}.
\end{proof}

\noindent While \eqref{eq:csin_aci} and \eqref{eq:chin_aci} provide an exact analytical solution for the concentrations of $\Hplus$ and $\Sub$ in an \ac{SVS}, the integral in \eqref{eq:variedc} cannot be solved in closed form and has to be evaluated numerically. Therefore, we derive a closed-form approximation that does not require numerical integration in the following.

\subsection{Closed-Form Approximation}
\label{sec:closed}
To obtain an approximate closed-form solution for $\Chin{t}$ and $\Csin{t}$, i.e., an analytical solution that circumvents the numerical integration in \eqref{eq:variedc}, we approximate the Michaelis-Menten term in \eqref{eq:ihsym} by linearization. This leads to the approximation of the fluxes of $\Sub$ and $\Hplus$ through the release module in \eqref{eq:ihsym}, i.e., 

\begin{equation}
    \issym{t} \approx 
        \gssym \mathds{1}_{\mathbb{R}^{+}}( \Csin{t})\mathds{1}_{[\ChinSwitch,\infty)}(\Chin{t}) := \gssym(t) \mathds{1}_{[\ChinSwitch,\infty)}(\Chin{t}),
    \label{eq:linear}
\end{equation}

\noindent where $\mathbb{R}^{+}$ and $\gssym(t) = \gssym \mathds{1}_{\mathbb{R}^{+}}( \Csin{t})$ denote the set of positive real numbers and the time-dependent $\Sub$ transport rate, respectively. Approximation \eqref{eq:linear} is justified as long as $\CsinZERO \gg \Km$ and assumes that the symporters operate with maximum rate as long as $\Sub$ is available and stop transporting when no $\Sub$ is left inside the vesicle. After inserting \eqref{eq:linear} into \eqref{eq:csin} and solving for $\Csin{t}$, \eqref{eq:csin_aci} simplifies to
\begin{equation}
    \Csin{t} =
    \begin{cases}
        \CsinStart{t} - \frac{\gssymT{t}}{\Vin} \left[t - \tSecStart{t} \right], & \text{if } \tssym{i} < t \le \tesym{i}, \\
         \CsinStart{t}, & \text{if } \tesym{i -1} < t \le \tssym{i}.
    \end{cases}
    \label{eq:csin_ana}
\end{equation}

\noindent Similarly, \eqref{eq:chin_aci} simplifies to
\begin{equation}
    \Chin{t} = 
a^{-1}b' - \left[ \ChinStart{t} -  a^{-1}b'\right]\mathrm{e}^{-a(t - \tSecStart{t})},
\label{eq:ana_chin}
\end{equation}

\noindent where $b' = b + \jhsymb \mathds{1}_{\{\tssym{i}, \tep{i}\}}(\tSecStart{t})$ with $\jhsymb = - \ghsymT{t}/\Vin$. In contrast to \eqref{eq:csin_aci} and \eqref{eq:chin_aci}, the closed-form approximations in \eqref{eq:csin_ana} and \eqref{eq:ana_chin} can be used to determine signal parameters such as the symporter start and end times, and the time of $\Sub$ depletion. This is required for an adequate analytical description of the $\Hplus$ and $\Sub$ concentrations (upon external illumination of the \ac{ND}) and the estimation of the duration of functionality of the \ac{ND}. 

The validity of the analytical solutions in \eqref{eq:csin_aci} and \eqref{eq:chin_aci} and the closed-form approximations in \eqref{eq:csin_ana} and \eqref{eq:ana_chin} will be verified by comparison to the numerical solution of \acp{ODE} \eqref{eq:chin} and \eqref{eq:csin} using a \ac{FDM} \cite{Grossmann2007}. As \ac{FDM} is exact for sufficiently small time steps $\step t$, we will consider the \ac{FDM} simulations as the ground truth in \cref{sec:sim}.

\subsection{Calculation of Cycle Phase Limits}
\label{sec:limits}

As mentioned in \cref{sec:func} and shown in \cref{fig:cycle1}, the end times of the phases in cycle $i$ are defined by $\tji{j}{i}$ for $j \in \{1,2,3,4\}$. The times $\tsp{i}$ and $\tep{i}$ for the start and the end of the illumination are defined by the external light signal $l(t)$. In contrast, the symport start and end times, $\tssym{i}$ and $\tesym{i}$, have to be calculated from the preceding cycle phases, i.e., phases (P2) and (P4), respectively. As \eqref{eq:chin_aci} is not invertible, $\tssym{i}$ and $\tesym{i}$ cannot be calculated from the exact analytical solution. Instead, the closed-form approximation \eqref{eq:ana_chin} can be inverted for all cycle phases, leading to approximate symport start and end times as follows:

\begin{equation}
    \tilde{t}^{\,(x)}_{i} = -a^{-1} \left[\log \left(\ChinSwitch - a^{-1}b' \right) - \log \left( \ChinStart{\tsymX{x - 1}{i}} - a^{-1}b' \right)\right] + \tsymX{x - 1}{i},
    \label{eq:tssymtheo}
\end{equation} 

\noindent for $x \in \{2,4\}$. To ensure that all $\tji{j}{i}$ are monotonically increasing, we have to apply correctional clipping for cycle types (b) and (c), for which some cycle phases do not occur, i.e., where some phases have a duration of $\SI{0}{\s}$ (see \cref{fig:cycle1} and \cref{sec:cycle_types}). Upon correctional clipping the actual symport start times $\tsymX{2}{i}$ and end times $\tsymX{4}{i}$ can be obtained for all possible cycle types as follows

\begin{equation}
    \tsymX{2}{i} = \max ( \min(\tsymX{1}{i}, \tilde{t}^{\,(2)}_{i} ), \tsymX{3}{i} ), \qquad  \tsymX{4}{i} = \max ( \min(\tsymX{3}{i}, \tilde{t}^{\,(2)}_{i} ), \tsymX{1}{i + 1} ).
    \label{eq:tsesym}
\end{equation}

\noindent For example, in case of a cycle of type (b) (see \cref{fig:cycle1}b), the symport start time, $\tssym{i}$, will be set to $\tep{i}$ according to \eqref{eq:tsesym} as the illumination time is too short for the symporters to become active. Similarly, $\tesym{i}$ is set to $\tsp{i + 1}$ because it would otherwise yield a physically infeasible value.

\subsection{Influence of Buffer}
\label{sec:buffer}

We assume that the system is immersed in a buffer solution (see \cref{sec:ass}) with total buffer molarity $\CZERO$ in both volumes $\VIN$ and $\VOUT$. Equation \eqref{eq:hh-eq} can be used to calculate the pH of a monoprotic buffer using the concentration of acid and base molecules and, thus, allows for an explicit buffer modeling. Therefore, we incorporate \eqref{eq:hh-eq} into the numerical \ac{FDM} solution as the ground truth for the impact of the buffer.

\noindent However, \eqref{eq:hh-eq} is not amenable to analytical solution as its incorporation into \eqref{eq:chin} and \eqref{eq:csin} would lead to an intractable system of \acp{ODE}. Thus, we approximate the effect of the buffer solution as an attenuation of the $\Hplus$ flux from one volume to another as proposed in \cite{Zifarelli2008}. This approach simply scales the flux rates of $\Hplus$, i.e., $\ghl$, $\ghp$, and $\ghsym$, by a factor $\attSec{t} = \ka \CZERO (\Chin{t} + \ka)^{-2}$ \cite[Eq. (12)]{Wagner1994}. This attenuation factor depends on the inner $\Hplus$ concentration $\Chin{t}$ and therefore varies over time. Incorporating the time-variant attenuation factor $\attSec{t}$ into \eqref{eq:chin} and \eqref{eq:csin} would thus lead to a system of non-linear \acp{ODE}. To avoid this, we assume that during \textit{each cycle phase} the attenuation factor $\attSec{t}$ remains constant and is computed using the initial intravesicular $\Hplus$ concentration of the cycle phase, $\ChinStart{t} = \Chin{\tSecStart{t}}$, i.e.,

\begin{equation}
    \attSec{\tSecStart{t}} \approx \ka \CZERO(\Chin{\tSecStart{t}} + \ka)^{-2}.
    \label{eq:attSec}
\end{equation}

\noindent Note that only values $\attSec{\tSecStart{t}} > 1$ are valid as other values correspond to an unbuffered scenario, which does not require flux attenuation. Scaling the $\Hplus$ fluxes in \eqref{eq:chin} and \eqref{eq:csin} with \eqref{eq:attSec} leads to a system of \acp{ODE} with a tractable solution for all cycle phases. 
In fact, the obtained solution is similar to \eqref{eq:csin_ana} and \eqref{eq:ana_chin} and simply uses rescaled auxiliary variables $j^{x*}_{Y} = j^{x}_{Y}\attSec{\tSecStart{t}}^{-1}$ for $x \in \{a, b\}$ and $Y \in \{\mathrm{L}, \mathrm{P}, \mathrm{Sym}\}$.
In our simulations, we will validate approximation \eqref{eq:attSec} by comparison to our numerical \ac{FDM} solution, where the impact of the buffer is explicitly modelled using \eqref{eq:hh-eq}.

\section{Extension to Multiple Vesicle System}
\label{sec:extension_mvs}

In this section, we expand the previous analyses of the \ac{SVS} to an \ac{MVS} consisting of $\nves$ \acp{ND}, where we consider the random distribution of several vesicle parameters. Then, we analyze the obtained \ac{MVS} model and investigate its behavior in comparison to the \ac{SVS}.

\subsection{Parameter Randomness}
\label{sec:mvs_analysis}

\begin{figure}
    \centering
    \includegraphics[width=0.8\linewidth]{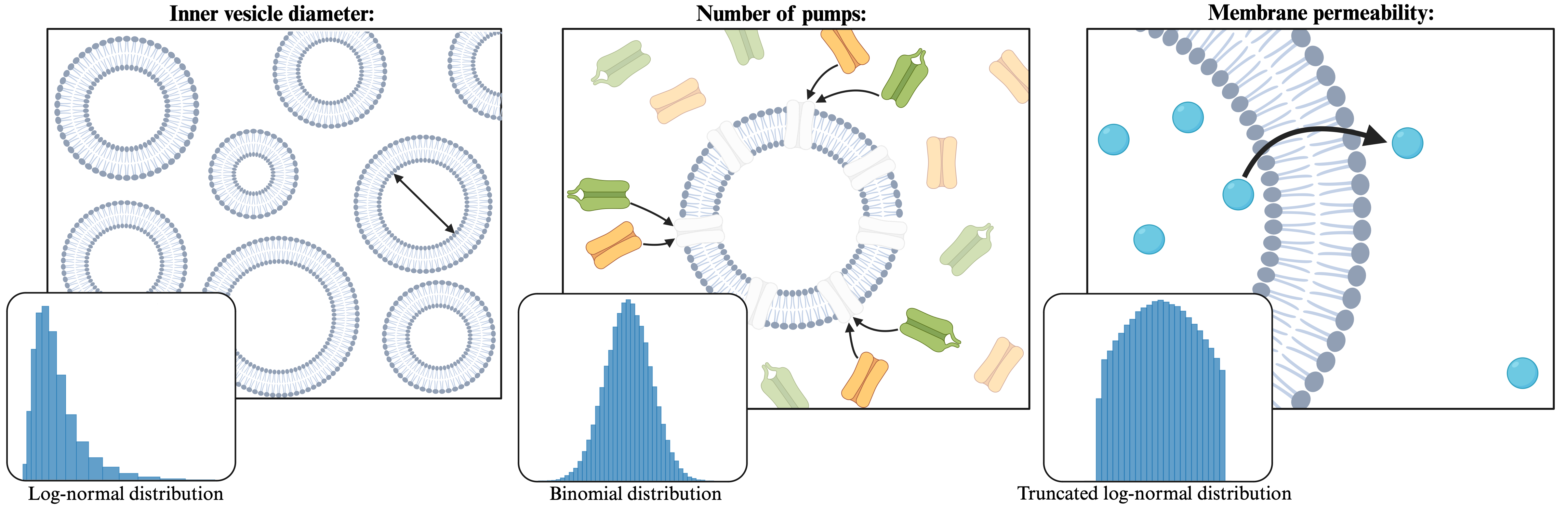}
    \vspace{-1.5em}
    \caption{Visualization of the considered random processes and distributions underlying the vesicle production process and influencing the vesicle diameter (left), the number of pumps and symporters per vesicle (center), and the membrane permeability to $\Hplus$ (right). Created with BioRender.com.}
    \label{fig:rand_proc}
\end{figure}

\noindent In \acp{MVS}, multiple vesicles are active and contribute to the release of \acp{SM} (see \cref{fig:sys_model}a). In realistic experiments, different vesicle parameters such as the vesicle diameter and the number of membrane proteins per vesicle are not deterministic but randomly distributed \cite{Poschenrieder2017}. This parameter randomness is unavoidable and mostly introduced during the vesicle production process \cite{Huang2017}. We discuss three practically relevant parameters and their distributions (see \cref{fig:rand_proc}) in the following\footnote{Note that if the parameters of all $\nves$ vesicles within an \ac{MVS} were deterministic and identical, the \ac{MVS} would collapse to an \ac{SVS} according to \eqref{eq:csouttot}.}.

\paragraph{Inner Vesicle Diameter} 
\noindent The distribution of the inner vesicle diameter, $\dinM{m}$, can be inferred from experiments and is illustrated on the left-hand side of \cref{fig:rand_proc}. As the vesicle membrane consists of lipids or polymers with characteristic lengths, the membrane thickness, $\dmem$, is typically considered to be deterministic. Therefore, variations of $\dinM{m}$ lead to different vesicle sizes. The actual diameter of the vesicles varies and it can be shown experimentally that the inner vesicle diameter of synthetic liposomes and polymersomes follows a shifted Log-normal distribution \cite{Tenchov1986}.
\noindent Consequently, to account for this randomness of the inner vesicle diameter, $\dinM{m}$, we assume 

\begin{equation}
    \dinM{m} \sim \Lognormal{x; l_{\mathrm{ves}}, \sigma_{\mathrm{ves}}, \mu_{\mathrm{ves}}}, \; f_{\mathcal{LN}}(x, \sigma_{\mathrm{ves}}, \mu_{\mathrm{ves}}) = \frac{\mathrm{e}^{-(\ln(x - l_{\mathrm{ves}}) - \mu_{\mathrm{ves}})^{2})/2\sigma_{\mathrm{ves}}^{2}}}{(x - l_{\mathrm{ves}}) \sigma_{\mathrm{ves}} \sqrt{2\pi}} \mathds{1}_{(l_{\mathrm{ves}}, +\infty)}\{x\},
    \label{eq:din_rand}
\end{equation}

\noindent for $x > l_\mathrm{ves}$, where $\Lognormal{x; l_{\mathrm{ves}}, \sigma_{\mathrm{ves}}, \mu_{\mathrm{ves}}}$ denotes the Log-normal distribution with shaping parameter $\sigma_{\mathrm{ves}} \in \mathbb{R}^{+}$, location parameter $\mu_{\mathrm{ves}} \in \mathbb{R}$, shift parameter $l_{\mathrm{ves}} \in \mathbb{R}$, and probability density function $f_{\mathcal{LN}}(\cdot)$. In \eqref{eq:din_rand}, $\ln(\cdot)$ denotes the natural logarithm. Note that $\sigma_{\mathrm{ves}}$ and $\mu_{\mathrm{ves}}$ correspond to the standard deviation and mean of the underlying Normal distribution of random variable $\ln(\dinM{m})$, respectively. Parameters $\sigma_{\mathrm{ves}}$ and $\mu_{\mathrm{ves}}$ have to be obtained from experimental data and depend on various factors such as the lipid or polymer chosen for the membrane and the vesicle production process.

\paragraph{Number of Pumps and Symporters}
\noindent Another source of randomness is the number of pumps and symporters per vesicle, whose ratio is expected to vary between individual \acp{ND} in \acp{MVS} (see center panel of \cref{fig:rand_proc}). However, to the best of our knowledge, there is no detailed description for the inherent randomness of protein reconstitution, yet. Therefore, we assume that the total number of proteins, $\ntotM{m} = \lfloor \pi (\dinM{m} + \dmem)^2 \rho_{\mathrm{ves}} \rfloor$, where $\lfloor\cdot\rfloor$ is the floor function yielding the greatest integer lower than or equal to its argument, depends linearly on the surface area of vesicle $m$ and the average density of proteins on the surface of a vesicle $\rho_{\mathrm{ves}}$, where density $\rho_{\mathrm{ves}}$ can be inferred from experimental data. There are thus $\ntotM{m} = \npM{m} + \nsymM{m}$ slots on the surface of each vesicle that may contain either a pump or a symporter, i.e., $\npM{m}$ and $\nsymM{m}$ are dependent random variables. Therefore, we assume that the number of pumps is given by a Binomial distribution

\begin{equation}
    \npM{m} \sim \Binom(\ntotM{m}, p_{\mathrm{P}}), \; f_{\Binom}(\ntotM{m}, p_{\mathrm{P}}) = \binom{\ntotM{m}}{k} p_{\mathrm{P}}^{k}(1- p_{\mathrm{P}})^{\ntotM{m} - k}, \; \nsymM{m} = \ntotM{m} - \npM{m},
    \label{eq:np_rand}
\end{equation}

\noindent where $\Binom(k; \ntotM{m}, \probPump)$ denotes the Binomial distribution with probability density function $f_{\Binom}$ for $\ntotM{m}$ samples and probability $\probPump \in [0,1]$, and $\dinM{m}$ is sampled according to \eqref{eq:din_rand}.

\paragraph{Membrane Permeability}
\noindent The third practically relevant source of randomness we consider is the permeability coefficient of $\Hplus$, as shown on the right-hand side of \cref{fig:rand_proc}. Recently, single vesicle analyses have shown that the permeability coefficient of lipid membranes with respect to $\Hplus$, $\ghlHatM{m}$, is not deterministic but rather randomly distributed over several orders of magnitude \cite{Guha2021}. For some lipids the coefficients seem to be correlated with the vesicle diameter, while for others such dependency cannot be observed. As these analyses are tedious to carry out, they cover only low numbers of samples. Thus, it is difficult to estimate the actual random distribution of permeability coefficients $\ghlHatM{m}$ from the available experimental data. We therefore assume that $\ln (\ghlHatM{m})$ is sampled from a truncated normal distribution (see right-hand side of \cref{fig:rand_proc}) where the truncation ensures that the sampled values remain in a physically feasible range. We assume that $\ln (\ghlHatM{m})$ is statistically independent from $\dinM{m}$ and therefore obtain

\begin{equation}
    \ln (\ghlHatM{m}) \sim \Truncnormal{x; \mu_{\mathrm{L}}, \sigma_{\mathrm{L}}, l_{\mathrm{L}}, u_{\mathrm{L}}}, \; f_{\mathcal{TN}}(x) = \begin{cases}
        \frac{\sqrt{2}}{\sigma_{\mathrm{L}}\sqrt{\pi}} \frac{\mathrm{exp}(-\frac{(x - \mu_{\mathrm{L}})^{2}}{2\sigma_{\mathrm{L}}^{2}})}{\mathrm{erf}(\frac{u_{\mathrm{L}} - \mu_{\mathrm{L}}}{\sigma_{\mathrm{L}} \sqrt{2}}) - \mathrm{erf}(\frac{l_{\mathrm{L}} - \mu_{\mathrm{L}}}{\sigma_{\mathrm{L}} \sqrt{2}})}, & \text{if } l_{\mathrm{L}} < x < u_{\mathrm{L}}, \\
        0, & \text{otherwise,}
    \end{cases}
    \label{eq:ghlM_dist}
\end{equation}

\noindent where $\Truncnormal{x; \mu_{\mathrm{L}}, \sigma_{\mathrm{L}}, l_{\mathrm{L}}, u_{\mathrm{L}}}$ denotes the truncated Normal distribution with probability density function $f_{\mathcal{TN}}$. Here, $\mu_{\mathrm{L}}$, $\sigma_{\mathrm{L}}$, $l_{\mathrm{L}} \in \mathbb{R}$, and $u_{\mathrm{L}} \in \mathbb{R}$ (with $u_{\mathrm{L}} > l_{\mathrm{L}}$) are the mean and standard deviation of the underlying non-truncated Normal distribution and the lower and upper truncation limits, respectively. Therefore, $\ghlHatM{m}$ is distributed in a truncated Log-normal fashion according to \eqref{eq:ghlM_dist}. In \eqref{eq:ghlM_dist}, $\mathrm{erf}(\cdot)$ denotes the Gaussian error function.

\subsection{Mathematical Analysis of \ac{MVS}}
\label{sec:ana_mvs}
\noindent In the proposed \ac{MVS} model, individual vesicles with randomly varying parameters are modeled as \acp{SVS} according to the assumption of \ac{SVS} independence (A4). To validate (A4), we utilize the extravesicular $\Hplus$ concentration in $\VOUTtot$ in the \ac{FDM} simulations, i.e., we average the $\ChoutM{m}{t}$ of all \acp{SVS}, i.e., $\ChoutTot{t} = 1/\nves \sum_{m=1}^{\nves} \ChoutM{m}{t}$. The \ac{FDM}, thus, serves as a baseline for our analytical approximations which are based on the assumption that $\VOUTtot$ is a well-mixed medium.

In order to characterize the behavior of an \ac{MVS}, we provide the following theorems.

\begin{theorem}
    \normalfont Given an \ac{MVS} that contains $\nves > 1$ vesicles where the parameter $\din$, $\np$, $\nsym$, and $\ghlHat$ are randomly distributed with a variance larger than 0, its behavior cannot be modeled by an \ac{SVS} with the mean parameter values $\dinMean$, $\npMean$, $\nsymMean$, and $\ghlHatMean$.
    \label{theo:t1}
\end{theorem}

\begin{proof}
    See \cref{sec:proof_mvs}.
\end{proof}

If we consider multiple \acp{MVS} (e.g., representing different experimental batches) with the same underlying random distributions of $\np$, $\din$, and $\ghlHat$, we can calculate the unbiased estimate of the inter-experiment variance among the \acp{MVS}. This variance indicates how similar different experimental batches behave, i.e., how reproducible an experiment is. Therefore, we define the inter-experiment variance for a signal, e.g., the extravesicular $\Sub$ concentration, as follows

\begin{equation}
	\mathrm{Var}\{\Csout{t}\} = \frac{1}{\nEx - 1} \sum_{q = 1}^{\nEx} \left(\CsoutQ{q}{t} - \frac{1}{\nEx} \sum_{q' = 1}^{\nEx}\CsoutQ{q'}{t} \right)^{2}
	\label{eq:interves_var}
 \end{equation}

\noindent where variable $q$ indicates the $q$-th experiment. In \eqref{eq:interves_var}, $\frac{1}{\nEx}\sum_{q' = 1}^{\nEx} \CsoutQ{q'}{t}$ constitutes the inter-experiment average. We make the following observation regarding the inter-experiment variance:

\begin{theorem}
\normalfont Given multiple \acp{MVS} with the same underlying distributions of $\din$, $\np$, and $\ghlHat$, $\mathrm{Var}\{\Csout{t}\}$ in \eqref{eq:interves_var} approaches 0 as the number of vesicles in the \ac{MVS} tends to infinity ($\nves \rightarrow \infty$).
\label{theo:t2}
\end{theorem}

\begin{proof}
Given (A5), $\CsoutQ{q}{t}$ is the empirical average over $\CsoutM{m,q}{t}$ per experiment $q$. As $\nves$ tends to infinity, we obtain $\lim_{\nves \rightarrow \infty} \CsoutQ{q}{t} = \mathcal{E}\{\CsoutM{m}{t}\}$. As the vesicle parameters are sampled from the same distributions among all experiments, $\mathcal{E}\{\CsoutM{m}{t}\}$ is independent of $q$. Consequently, the argument in the sum over $q$ in \eqref{eq:interves_var} becomes 0 and so does $\mathrm{Var}\{\Csout{t}\}$.
\end{proof}

\FloatBarrier

\section{Characterization of \ac{ND}-based \ac{MC} Transmitters}
\label{sec:sim}

\noindent In this section, the results for both the \ac{SVS} and the \ac{MVS} obtained for the exact analytical solution \eqref{eq:csin_aci} and \eqref{eq:chin_aci} and the closed-form approximation \eqref{eq:csin_ana} and \eqref{eq:ana_chin} are presented and compared to the ground truth, which is numerically obtained by \ac{FDM}. First, we investigate the impact of the buffer molarity on the dynamics of the energizing module of an \ac{SVS}. Then, the functionality of the entire \ac{SVS} is examined for varying ratios of the numbers of pumps and symporters. As the transport rate of proteins cannot be changed straightforwardly, the numbers of pumps and symporters in the vesicle membrane are important design parameters since they directly scale the fluxes of $\I$ and $\Sub$ (see \eqref{eq:flux_p} and \eqref{eq:ihsym}). Lastly, for the \ac{SVS}, the influence of different transport rates of symporters (corresponding to different symporter types) and different membrane permeabilities on the release of $\Sub$ for different illumination durations is discussed.
\\ \noindent For the analysis of \acp{MVS}, the differences among the vesicles within one experiment are studied in comparison to the joint behavior of an \ac{MVS} over several experimental realizations. Then, a sensitivity analysis is conducted to determine critical experimental parameters that heavily influence the system behavior upon variation.
\\ \noindent If not specified otherwise, all simulations are conducted using the default parameters shown in \Cref{tab:params}. These default values are chosen to be in line with experimental data, if available. The time step $\step t$ is relevant for the numerical \ac{FDM} baseline, which requires time discretization.

\begin{table}[!ht]
    \centering
    \caption{Default parameters for the evaluations.}
    \vspace{-1em}
    \resizebox{\columnwidth}{!}{
    \begin{threeparttable}
    \begin{tabular}{|l|r@{\hspace{1pt}}l|c|}
        \hline
        \textbf{Parameter} & \multicolumn{2}{c|}{\textbf{Value}} & \textbf{Ref.}\\
        \hline
        \hline
        $\step t$ & $\num{1e-2}$ & $\si{\second}$ &  \\ \hline
        $\CsinZERO$ & \num{300} & \si{\mol\per\cubic\meter} & $\ast$\\ \hline
        $\CZERO$ & \num{20} & \si{\mol\per\cubic\meter} & \cite{Goers2018} \\ \hline
        $\ka$ & $\num{6.2e-5}$ & \si{\mol \per \cubic \meter} & \cite{Ellison1958} \\ \hline
        $\din$ & $\num{87}$ & $\si{\nano\meter}$ & \cite{Rideau2018} \\ \hline
        $\dmem$ & $\num{14}$ & $\si{\nano\meter}$ & \cite{Rideau2018} \\ \hline
        $\Vouttot$ & $\num{1e-6}$ & $\si{\cubic\meter}$ & \cite{Goers2018} \\ \hline
        $\ChinZERO$ & $\num{3.98e-5}$ & $\si{\mol\per\cubic\meter}$ &  $\ast\ast$\\ \hline
        $\ChoutZERO$ & $\num{3.98e-5}$ & $\si{\mol\per\cubic\meter}$ &  $\ast\ast$\\ \hline
        $\Km$ & $\num{1.3e-2}$ & $\si{\mol\per\cubic\meter}$ & \cite{Levy1998} \\ \hline
    \end{tabular}
    \end{threeparttable}
    \begin{threeparttable}
    \begin{tabular}{|l|r@{\hspace{1pt}}l|c|}
        \hline
        \textbf{Parameter} & \multicolumn{2}{c|}{\textbf{Value}} & \textbf{Ref.}\\
        \hline
        \hline
        $\fSym$ & $3$ & & \cite{Ryan2009} \\ \hline
        $\gssymHat$ & $\num{0.006}$ & $\si{\per\second}$ & \cite{Tubbe1992}\\ \hline
        $\nsym$ & $30$ & & \cite{Goers2018} \\ \hline
        $\ghl$ & $\num{3e-6}$ & $\si{\meter \per \second}$ & \cite{Deamer1987} \\ \hline
        $\ghpHat$ & $\num{0.03}$ & $\si{\per \second}$ & \cite{Dioumaev2003} \\ \hline
        $\np$ & $40$ & & \cite{Goers2018} \\ \hline
        $\thresh$ & 0.015 & & $\dag$ \\ \hline
        $\nves$ & $10^{11}$ & & \cite{Gaitzsch2019} \\ \hline
        $\nMod$ & 100 & & \\ \hline
        $\nEx$ & 10 & & \\ \hline
    \end{tabular}%
    \end{threeparttable}
    \begin{threeparttable}
    \begin{tabular}{|l|r@{\hspace{1pt}}l|c|}
        \hline
        \textbf{Parameter} & \multicolumn{2}{c|}{\textbf{Value}} & \textbf{Ref.}\\
        \hline
        \hline
        $\mu_{\mathrm{ves}}$ & 4.16 & $\log(\si{\nano\meter})$ & $\ddag$ \\ \hline
        $\sigma_{\mathrm{ves}}$ & 0.62 & & $\ddag$ \\ \hline
        $l_{\mathrm{ves}}$ & 39.74 & \si{\nano\meter} & $\ddag$ \\ \hline
        $\rho_{\mathrm{tot}}$ & 0.168 & \si{\per\square\nano\meter} & \\ \hline
        $p_{\mathrm{P}}$ & 4/7 & & \\ \hline
        $\mu_{\mathrm{L}}$ & -5.52 & & \\ \hline
        $\sigma_{\mathrm{L}}$ & 0.25 & & \\ \hline
        $l_{\mathrm{L}}$ & -5.77 & & \\ \hline
        $u_{\mathrm{L}}$ & -5.27 & & \\ \hline
        \multicolumn{4}{c}{}
    \end{tabular}
    \end{threeparttable}}
    \label{tab:params}
    \smallskip
\begin{minipage}{\linewidth}\footnotesize
\begin{itemize}[leftmargin=0.5cm, itemsep=-4pt, topsep=16pt]
\vspace{2pt}
    \item[$\ast$] Expected concentration that can be loaded into vesicles based on reported, similarly high encapsulation concentrations \cite{Nallani2006}. 
   \item[$\ast\ast$] Corresponds to a pH of 7.4, which is standard for some commonly used buffers \cite{Goers2018}.
   \item[$\dag$] Rough estimate for a symporter that is active at the resting membrane potential of a cell.
   \item[$\ddag$] Value is obtained from curve fitting on experimental data (refer to \cref{ssec:mvs_results} for details).
\end{itemize}
\end{minipage}
\vspace{-1em}
\end{table}

\subsection{Energizing Module: Evaluation of the Impact of the Buffer}
\label{sec:mol}

\noindent In order to assess the functionality of the proposed $\Hplus$-based energizing module under varying experimental conditions, in our first evaluation, we consider an \ac{ND} without release module, i.e., $\nsym = 0$, for different buffer molarities $\CZERO$. \Cref{fig:mol} shows $\Chin{t}$ as a function of time $t$ for various buffer concentrations $\CZERO$ as obtained from the exact analytical solution \eqref{eq:chin_aci} (green), the approximate solution \eqref{eq:ana_chin} (orange), and the numerical baseline (blue). The results were obtained for $\SI{600}{\s}$ of continuous illumination followed by $\SI{600}{\s}$ without illumination, as indicated by the shading in \cref{fig:mol}. We observe that the general signal shapes are in agreement with experimental data from the literature employing a similar setup using light-driven $\Hplus$ pumps (compare \cref{fig:mol} with \cite[Fig. 3]{Harder2024}). In both our simulations and the experimental measurements, illumination causes an increase of $\Chin{t}$ (i.e., a decrease in the intravesicular pH) that approaches saturation exponentially. Dark phases similarly cause a return to the initial value of $\Chin{t}$. This suggests that our developed system model successfully captures the behavior of the envisioned \ac{ND}. \Cref{fig:mol} also highlights the importance of modeling buffering effects, as the system dynamics of a buffered system ($\CZERO > \SI{0}{\mol\per\cubic\meter}$) clearly differ from those of an unbuffered system ($\CZERO = \SI{0}{\mol\per\cubic\meter}$). We observe that a higher buffer molarity causes a smaller slope of $\Chin{t}$ during both illumination ($l(t) = 1$) and no illumination ($l(t) = 0$). As expected, for increasing buffer molarity, the $\Hplus$ in- and outfluxes are more attenuated by the buffer and therefore the rate of change in $\Chin{t}$ is smaller. It is noteworthy that, for all buffer molarities, $\Chin{t}$ approaches the same value for sufficiently long illumination durations. 
This value is the dynamic equilibrium concentration $\ChinEq$ (black line in \cref{fig:mol}), where the influx $\ihp{t}$ caused by the pumps is equal to the outflux $\ihl{t}$ caused by the leakage. As all $\Hplus$ molecules entering or leaving $\VIN$ are equally buffered, $\ChinEq$ is unaffected by the buffer molarity. However, as can be observed in \cref{fig:mol}, the speed at which the equilibrium is reached changes, e.g., the curve for $\CZERO = \SI{100}{\mol\per\cubic\meter}$ does not reach $\ChinEq$ as the illumination period considered in \cref{fig:mol} is too short. Generally,  \cref{fig:mol} shows that the analytical solutions are in good agreement with the numerical baseline. Interestingly, the closed-form approximation is very accurate while entailing much lower computational costs compared to the other solutions. Only when the changes in $\Chin{t}$ are large in the buffered scenario, e.g., for low buffer molarities (e.g., $\CZERO = \SI{10}{\mol\per\cubic\meter}$), small deviations between the analytical solutions and the \ac{FDM} solution can be observed. 
The reason for these deviations is the phase-wise constant attenuation factor $\attSec{\tSecStart{t}}$ in \eqref{eq:attSec}. As long phases typically lead to substantial changes in $\Chin{t}$ due to $\Hplus$ in- or outflux, the assumption of constant $\Chin{t}$ used for calculating $\attSec{\tSecStart{t}}$ in \eqref{eq:attSec} becomes invalid. The deviation decreases again when $\ChinEq$ is reached as no $\Chin{t}$ changes   occur anymore as long as the cycle phase does not change.
Conclusively, we note that the buffer molarity of the system determines how responsive the energizing module is with respect to changes in the external stimulus. In particular, higher buffer molarities introduce a latency to the stimulus response, while low buffer molarities lead to a more responsive system.

\subsection{Energizing and Release Module: Evaluation of the \ac{ND} Functionality}

\begin{figure}
\centering
\begin{minipage}[b]{.47\textwidth}
  \centering
  \includegraphics[width=0.99\textwidth]{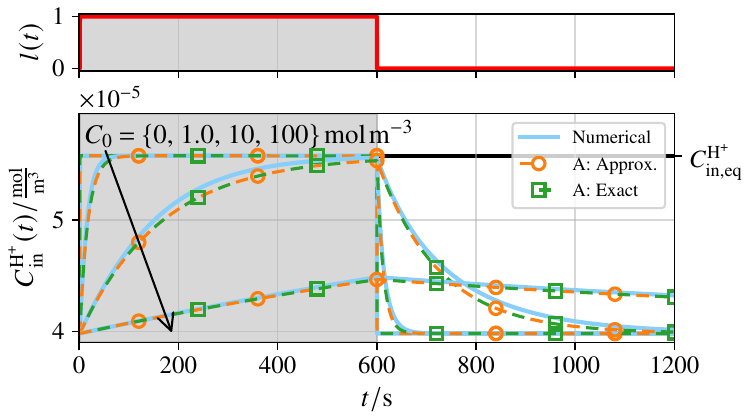}
  \vspace{-2.5em}
  \caption{Intravesicular $\Hplus$ concentration, $\Chin{t}$, (bottom) for one illumination period without release module, i.e., $\nsym = 0$, and for varying buffer molarities $\CZERO$. Results obtained with \ac{FDM} (blue), the exact analytical solution \eqref{eq:chin_aci} (green), and the approximate analytical solution \eqref{eq:ana_chin} (orange) are shown as a function of time $t$. The light signal $l(t)$ (red) is shown on the top. Shaded gray areas indicate times during which $l(t) = 1$. The black line shows the $\Hplus$ concentration, $\ChinEq$, where in- and outflux to/from the vesicle are in equilibrium while $l(t) = 1$.}
  \label{fig:mol}
\end{minipage}
\hspace{0.2em}
\begin{minipage}[b]{.47\textwidth}
  \centering
    \includegraphics[width=0.99\linewidth]{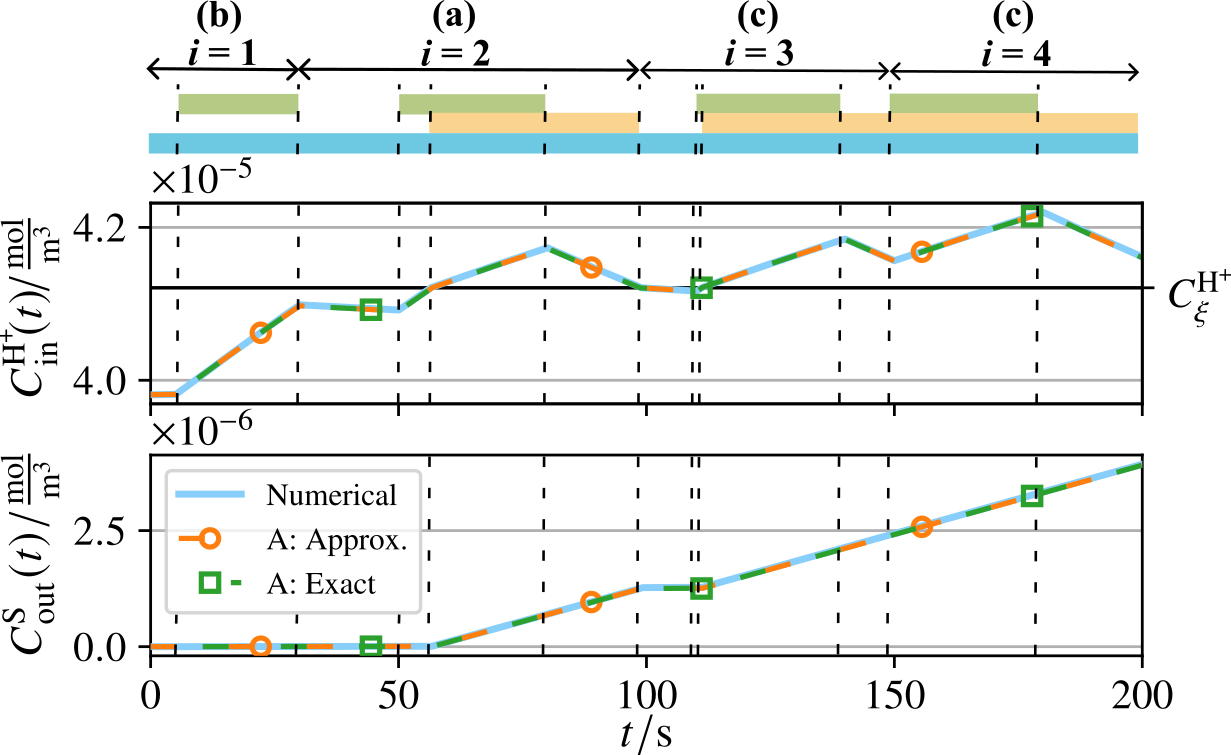}
    \vspace{-2.5em}
    \caption{Intravesicular $\Hplus$ concentration (middle panel) and extravesicular $\Sub$ concentration (bottom panel) for different cycle types (a), (b), and (c) and solution approaches. The active system components during each cycle phase are shown on top (green: proton pumps, orange: symporters, blue: leakage) and dashed lines indicate cycle phase limits. The intravesicular $\Hplus$ concentration required for symport activation, $\ChinSwitch$, is shown by the horizontal black line in the middle panel. 
    \\ \\ }
    \label{fig:types}
\end{minipage}%
\end{figure}

\paragraph{Accurate Analytical Solutions for all Cycle Types}

\noindent In comparison to \cref{fig:mol}, which showed a scenario without release modules, we now consider a system with $\nsym > 0$. \Cref{fig:types} shows the system response to a light signal chosen such that all three cycle types (a)--(c) occur (see \cref{fig:cycle1}). Evidently, for all cycles types, the analytical and numerical solutions are in good agreement. This indicates that our model adequately captures the system behavior for different light signal patterns. The first illumination phase in \cref{fig:types} is not long enough for $\ChinSwitch$ to be reached, thus the symporters remain inactive and a cycle of type (b) occurs. The symporter inactivity can also be observed in the amount of released $\Sub$ (bottom panel of \cref{fig:types}), which remains unchanged during $t < \SI{50}{\s}$. For future experiments, it is necessary to carefully design the light signal in such a way that cycles of type (b) (see \cref{fig:cycle1}b) are avoided, as the goal of \ac{SM} release is not achieved. The illumination phase of cycle $i = 2$ is of similar length to that of cycle $i = 1$, but due to a higher initial concentration $\ChinStart{t = \tsp{2}}$ when cycle $i = 2$ starts at $t = \SI{50}{\s}$, the symporters become active during this cycle. We observe that the slope of $\Chin{t}$ decreases and that $\Csout{t}$ increases linearly when the symporters are active, as $\isr{t} > 0$ for $\Chin{t} > \ChinSwitch$. During the subsequent dark phase ($t \in [80, 110]\si{\s}$), the symporters become inactive at $t \approx \SI{100}{\s}$ (see sudden stop of increase in $\Csout{t}$). Therefore, cycle $i = 2$ is a regular cycle of type (a). The following cycles $i = 3$ and $4$ are of type (c) as the dark phase ($t \in [140, 150]\si{\s}$) between the illumination phases is too short to stop the symporter activity (see continuous increase in $\Csout{t}$ for $t > \SI{110}{\s}$). \Cref{fig:mol} shows that our derived analytical solutions are in good agreement with the numerical baseline for all cycle types and are, therefore, expected to model the system at hand reliably.

\paragraph{Impact of the Composition of the Membrane Proteins over Time}

\noindent \Cref{fig:ratios} shows $\Chin{t}$, $\Csout{t}$, and the outflux of $\Sub$ caused by the release module, $\isr{t}$, for different ratios of $\np$ and $\nsym$, while the total number of membrane proteins $\ntot$ remains constant. Generally, we observe that a smaller number of symporters leads to larger $\Chin{t}$ during the illumination phases due to a lower symport-caused $\Hplus$ outflux. Similarly, the influence of a lower number of symporters is also observable from the smaller slope of $\Csout{t}$ (see center panel of \cref{fig:ratios}) or, equivalently, from the lower amplitude of $\isr{t}$ during illumination periods (bottom panel of \cref{fig:ratios}). However, the higher peaks of $\Chin{t}$ for smaller $\nsym$ lead to a longer symport duration, $\tesym{i} - \tssym{i}$, in each cycle as shown by the increasing width of the rectangles in $\isr{t}$ for decreasing $\nsym$. These observations lead to the conclusion that a lower number of symporters does not necessarily correlate with an overall lower amount of released $\Sub$ (which is proportional to the area of the rectangles in $\isr{t}$) as a lower outflux rate causes longer symport durations.
For the case $\np/\nsym = 1$, we also observe the effect of substrate depletion, i.e., $\Csin{t}$ approaches 0, in \cref{fig:ratios}. To investigate this phenomenon, we have chosen a low $\CsinZERO = \SI{3.14}{\mol\per\cubic\meter}$ for which substrate depletion takes place at around $t = \SI{6500}{\s}$. However, in practice, larger $\CsinZERO$ are achievable and should be used to increase the longevity of the \ac{ND} (see \Cref{tab:params}).
For $\np/\nsym = 4/3$, the substrate is depleted even earlier as evident from the fact that $\isr{t} = 0$ for all plotted times $t > \SI{6350}{\s}$ (highlighted in red in \cref{fig:ratios}). On the other hand, substrate depletion is not reached during the simulation time for $\np/\nsym = 3/4$. Consequently, the rectangular signal shape of $\isr{t}$ is maintained until the end of the simulation. We also note that during substrate depletion, the accuracy of the approximate solution \eqref{eq:csin_ana} decreases (see mismatch between the blue and orange curves in the bottom panel of \cref{fig:ratios} for $\np = \nsym$) as the linear approximation is not able to capture the decrease in the symport rate characteristic for Michaelis-Menten kinetics (see \eqref{eq:ihsym} for $\Csin{t} \ll \Km$). In contrast, the exact analytical solution \eqref{eq:chin_aci} is able to capture this effect, and the slight deviations from the numerical baseline are attributed to the finite time resolution in the numerical integration for obtaining $\alpha(t)$ in \eqref{eq:variedc}. 
Note that the light signal in \cref{fig:ratios} may be interpreted as a modulated transmit signal for concentration shift keying. Since the difference in $\isr{t}$ during illumination and dark phases mimics the shape of the optical transmit signal, the resulting signal may be suitable for encoding information. Generally, \cref{fig:ratios} shows that the ratio $\np/\nsym$ is an important design parameter of the \ac{ND}. For example, if the envisioned use case of the \ac{ND} requires a prolonged, sustained release of $\Sub$ upon illumination (wide rectangles), large $\np/\nsym$ should be chosen, while for shorter temporal responses to the external stimuli (narrow rectangles) small $\np/\nsym$ are favorable.

\subsection{Estimation of Substrate Release}
\noindent In order to utilize the proposed \ac{ND} as \ac{TX} for synthetic \ac{MC} for applications such as \ac{TDD} or information transmission, it is necessary to design the release of $\Sub$ adequately in consideration of the limited $\Sub$ resources inside the vesicle. One important variable in this context is the expected number of released $\Sub$ in response to a specific illumination duration, which scales linearly with the expected duration of symport during an illumination cycle (while $\Csin{t} \gg \Km$). The closed-form expression \eqref{eq:tssymtheo} allows for the estimation of the start $\tssymNoi$ and end $\tesymNoi$ times of the symport (see \cref{sec:limits}). To validate the results, we compare the values obtained from \eqref{eq:tssymtheo} with the numerically simulated symport duration. The top panel of \cref{fig:sym_dur} shows the symport duration, $\tesymNoi - \tssymNoi$, and the bottom panel shows the final extravesicular $\Sub$ concentration $\Csout{\tepNoi}$ for varying $\ghlHat$ and $\gssymHat$, which may correspond to different types of symporters and vesicle membranes, respectively. Note that a minimum illumination duration (denoted by the vertical black line in \cref{fig:sym_dur}) is required to reach $\ChinSwitch$ and to ensure that a cycle exhibits symporter activity. First, we observe that the symport duration increases approximately linearly during illumination after the minimum required illumination time. Moreover, we observe that the higher $\gssymHat$, the lower is the symport duration for a given illumination duration as the symporter-caused $\Hplus$ outflux is larger. Similarly, a higher leakage rate $\ghlHat$ (green curves in \cref{fig:sym_dur}) shortens the symport duration as it also causes a larger $\Hplus$ outflux. For a given $\gssymHat$, the total concentration of released $\Sub$ during the illumination phase, $\Csout{\tepNoi}$, is therefore lower for larger $\ghlHat$ (see bottom panel in \cref{fig:sym_dur}). In practice, $\ghlHat$ depends on the type of vesicle membrane (e.g., lipid or polymeric) and the choice of $\I$ and can vary substantially. Hence, its impact has to be considered carefully in experimental design. When both $\ghlHat$ and $\gssymHat$ are small, i.e., $\ghlHat = \SI{5e-6}{\meter\per\s}$ and $\gssymHat \le \SI{0.005}{\per\s}$ (blue curve in top panel), the curves for the symport duration do not exhibit linear behavior. Instead, the symport duration first increases quickly but then more slowly as the illumination duration increases. 
Moreover, we observe a slight deviation between our analytical approximation for the symport duration (dashed lines in \cref{fig:sym_dur}) and the duration obtained from \ac{FDM} (solid lines in \cref{fig:sym_dur}) for $\ghlHat = \SI{5e-6}{\meter\per\s}$ and $\gssymHat \le \SI{0.005}{\per\s}$. This deviation is caused by a large $\ChinEq$ for a low $\Hplus$ outflux which causes the analytical approximation for the buffer effect to deviate more substantially from its actual values (as described in \cref{sec:mol}). \Cref{fig:sym_dur} also shows that the leakage flux has a substantial impact on the symport duration, $\tesymNoi - \tssymNoi$, i.e., the responsiveness of the \ac{ND} to external stimuli, while the symport rate constant determines the strength of the chemical signal, i.e., $\Csout{t}$. These observations can guide the choice of co-transporters for the release module and the choice of the vesicle membrane.

\begin{figure}
\centering
\begin{minipage}[t]{.47\textwidth}
  \centering
    \centering
    \includegraphics[width=0.99\textwidth]{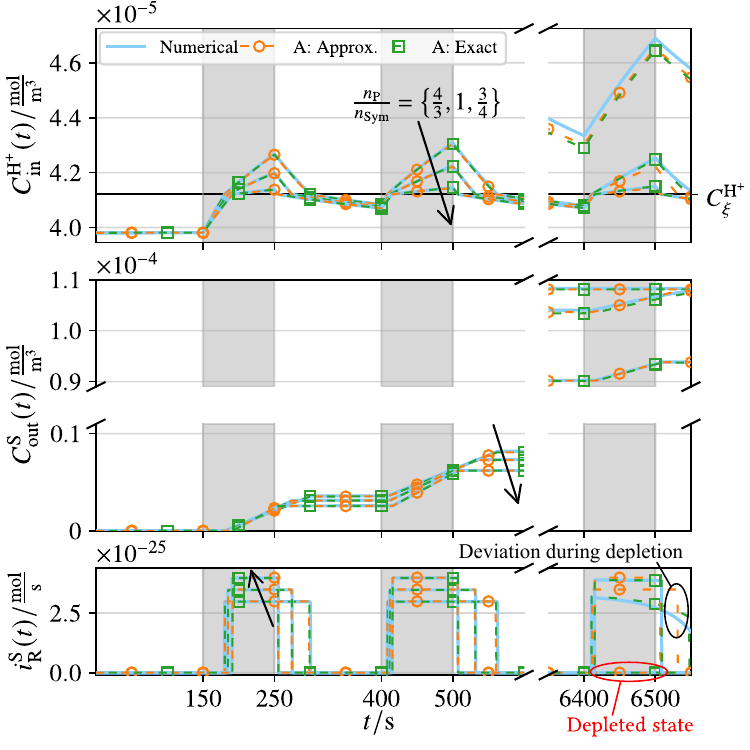}
    \caption{Intravesicular $\Hplus$ concentration, extravesicular $\Sub$ concentration, and the outflux of $\Sub$ caused by the symporters over multiple illumination cycles for varying protein ratios $\np / \nsym$ and $\CsinZERO = \SI{3.14}{\mol\per\cubic\meter}$. The arrows indicate decreasing $\np / \nsym$. Shaded gray areas indicate times during which $l(t) = 1$.}
    \label{fig:ratios}
\end{minipage}
\hspace{0.2em}
\begin{minipage}[t]{.47\textwidth}
    \centering
    \includegraphics[width=0.99\textwidth]{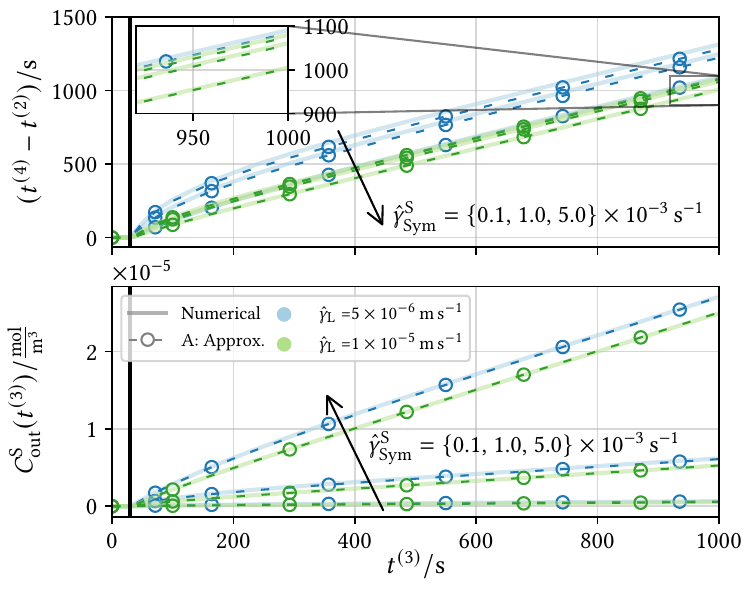}
    \caption{Top: Symport duration over the illumination duration ($\tspNoi = \SI{0}{\s}$) for different symport and leakage rates. Bottom: Corresponding $\Csout{t}$ during the illumination period. The vertical black line marks the minimum illumination time needed for symporter activity.}
    \label{fig:sym_dur}
\end{minipage}%
\end{figure}

\subsection{Multiple Vesicle System}
\label{ssec:mvs_results}

\noindent Now, we consider an \ac{MVS} consisting of $\nves$ \acp{ND}. Therefore, parameters $\din$, $\ghlHat$, $\np$, and $\nsym$ are randomly distributed as described in \cref{sec:mvs_analysis}, where the parameter values of the distributions are given in the right-most columns of \cref{tab:params}. \Cref{fig:din_exp_dist_fit} shows the experimentally measured distribution of the total vesicle diameter, i.e., $\dinM{m} + 2\dmem$, and the corresponding fitted log-normal distribution. It is apparent that the log-normal distribution (see \eqref{eq:din_rand} in \cref{sec:mvs_analysis}) with parameters $l_{\mathrm{ves}} = \SI{39.74}{\nano\meter}, \mu_{\mathrm{ves}} = 4.16\log(\mathrm{nm}),$ and $\sigma_{\mathrm{ves}} = 0.62$ is an adequate fit for the experimental data\footnote{Polymersomes were formed from the amphiphilic triblock copolymer poly(2-methyloxazoline)\textsubscript{15}-poly(dimethylsiloxane)\textsubscript{68}-poly(2-methyloxazoline)\textsubscript{15} (PMOXA\textsubscript{15}–PDMS\textsubscript{68}–PMOXA\textsubscript{15}). The production process was carried out in the miniaturized stirred-tank reactor system bioREACTOR 48 (2mag, Munich, Germany) with a working volume of \SI{12}{\milli\liter} \cite{Poschenrieder2016}. A volume of 600\,µL of a 20\,\% (weight/volume) ethanolic polymer solution was added to the buffer (\SI{20}{\mol\per\cubic\meter} potassium phosphate buffer, \SI{150}{\mol\per\cubic\meter} potassium chloride, pH 8) in a 1:20 ratio (resulting in a 1\,\% (weight/volume) polymer concentration) and then stirred at 4000\,rpm and \SI{30}{\degreeCelsius} with an S-shaped stirrer. The vesicle size distribution was determined by dynamic light scattering (ZetaSizer Nano series, Malvern Panalytics GmbH, Kassel, Germany). The fabrication process was terminated after one hour, when a polydispersity index of less than 0.2 was achieved.}. Interestingly, \cref{fig:din_exp_dist_fit} shows that there is a physical minimum of the total vesicle diameter of approximately $\SI{40}{\nano\meter}$ for the given data set. This minimum is likely governed by the vesicle production process, membrane thickness, and the interactions between the membrane molecules. Such a clear physical limit does not seem to exist for the maximum vesicle diameter, as the experimental data shows a heavy tail towards large vesicle diameters, leading to a skewed distribution. This behavior is adequately captured by a shifted log-normal distribution, which has a lower but no upper bound. \Cref{fig:gl_trunc_lognorm} shows the histogram of values for $\ghlHat$ sampled from a truncated log-normal distribution (see \eqref{eq:ghlM_dist}). Similarly, \cref{fig:np_binom_fit} shows the Binomial distribution and sampled random variables for the number of pumps $\np$ for a vesicle with $\din = \SI{117.67}{\nano\meter}$ (mean diameter from \cref{fig:din_exp_dist_fit}). 

\noindent For the following results, we distinguish between the behavior of different \acp{SVS}, i.e., \acp{ND}, and the behavior of several \ac{MVS}, i.e., experiments. \Cref{fig:interandintra} shows the variation between \emph{different \acp{ND}} within one experiment (top panels) and between the \acp{MVS} of \emph{different experiments} (bottom panels). The inter-vesicle and inter-experiment means are shown as solid and dashed lines in the top and bottom panels of \cref{fig:interandintra}, respectively. Similarly, the shaded area around the means shows one standard deviation for the inter-vesicle and -experiment results. The number of experiments was $\nEx = 10$ while the number of individually modeled vesicles per experiment was $\nMod = 100$. This means that, while we assume that our \ac{MVS} consists of $\nves = 10^{11}$ vesicles and divide $\VOUTtot$ accordingly, we only model a small fraction of these vesicles ($\nMod \ll \nves$) to limit computational complexity. The results from \ac{FDM} (blue curves) in \cref{fig:interandintra} were calculated using the total extravesicular $\Hplus$ concentration, $\ChoutTot{t}$, and serve as a baseline. According to (A4), the analytical approximation assumes total independence of all subsystems and therefore refrains from computing $\ChoutTot{t}$. Additionally, \cref{fig:interandintra} shows the expected behavior of an \ac{SVS} based on the mean parameter values for $\np, \nsym, \ghlHat$, and $\din$ of the \ac{MVS} (black dashed line).
From the top panels of \cref{fig:interandintra}, it can be observed that the inter-vesicle variance is substantial, as reflected by the large standard deviation of $\Chin{t}$ and $\Csout{t}$. This means that the behavior of individual \acp{ND} within one \ac{MVS} is diverse and inhomogeneous. Additionally, the mean responses of $\Chin{t}$ and $\Csout{t}$ are not equal to the response of a vesicle with average parameters. This observation suggests that careful consideration of the parameter distribution is required to adequately model the behavior of the vesicles that occur in chemical realizations. Similarly, the mean behavior of the \acp{MVS} (bottom panel of \cref{fig:interandintra}) also deviates from the corresponding mean \ac{SVS}, as expected from \cref{theo:t1}. In fact, as the influence of the random parameter $\dinM{m}$ results in squared and cubic changes in the intravesicular volume and the transport rates, the mean $\Csout{t}$ of an \ac{SVS} clearly lies below the inter-experiment mean of the \acp{MVS} (bottom right panel in \cref{fig:interandintra}). Notably, the inter-experiment variance is much less substantial than the inter-vesicle variance. As stated in \cref{theo:t2}, the inter-experiment variance tends to 0 as the number of vesicles per experiment tends to infinity. As the number of vesicles in experiments is expected to be around $10^{11} \si{\per \milli \liter}$, the inter-experiment variance is likely negligible in realistic setups.
We point out that, in \cref{fig:interandintra}, the differences between the \ac{FDM} baseline and our analytical approximation are similarly small as for the \ac{SVS}, where they result from the constant buffer attenuation factor. Consequently, for the \ac{MVS}, the deviations from the baseline in \cref{fig:interandintra} are also expected to (mostly) originate from the constant buffer attenuation rather than the violation of (A4). This is in line with our expectation that considering \ac{SVS} independence for the analytical approximation does not cause a substantial error when $\VoutM{m} \gg \VinM{m}$, as would be the case for most experimental settings. If $\VoutM{m} \gg \VinM{m}$, the fluxes from and to $\VOUTM{m}$ lead to small concentration changes in $\VoutM{m}$ due to its large volume size. Generally, the analysis of \acp{MVS} clearly indicates that considering multiple \acp{ND} within realistic \acp{TX} is crucial as their behavior deviates from the expected behavior of a single \ac{ND}. Nevertheless, the behavior of the overall \ac{MVS} is consistent among different experiments. This indicates that future wet lab implementations of the system will be reproducible and, thus, reliable.

\begin{figure}
\centering
\begin{minipage}[t]{.32\textwidth}
    \centering
    \includegraphics[width=0.99\textwidth]{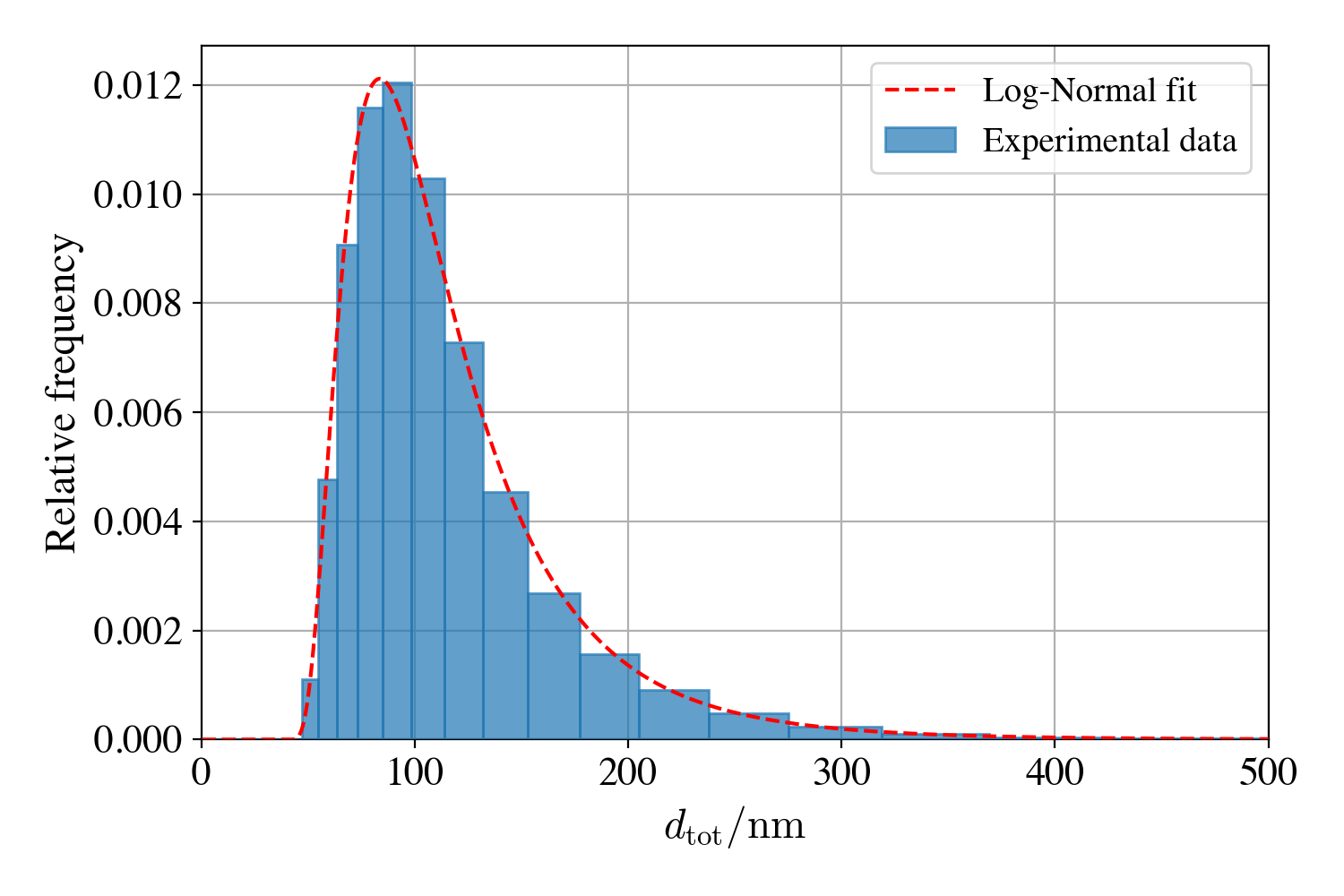}
    \caption{Experimentally measured $\dinM{m} + 2 \dmem$ for PMOXA\textsubscript{15}-PDMS\textsubscript{68}-PMOXA\textsubscript{15} \mbox{triblock co-polymersomes} synthesized in dichloroethan (blue bars) and a log-normal fit (red curve) with parameters $l_{\mathrm{ves}} = \SI{39.74}{\nano\meter}, \mu_{\mathrm{ves}} = 4.16\log(\mathrm{nm}),$ and $\sigma_{\mathrm{ves}} = 0.62$.}
    \label{fig:din_exp_dist_fit}
\end{minipage}
\hspace{0.1em}
\begin{minipage}[t]{.32\textwidth}
    \centering
    \includegraphics[width=0.99\textwidth]{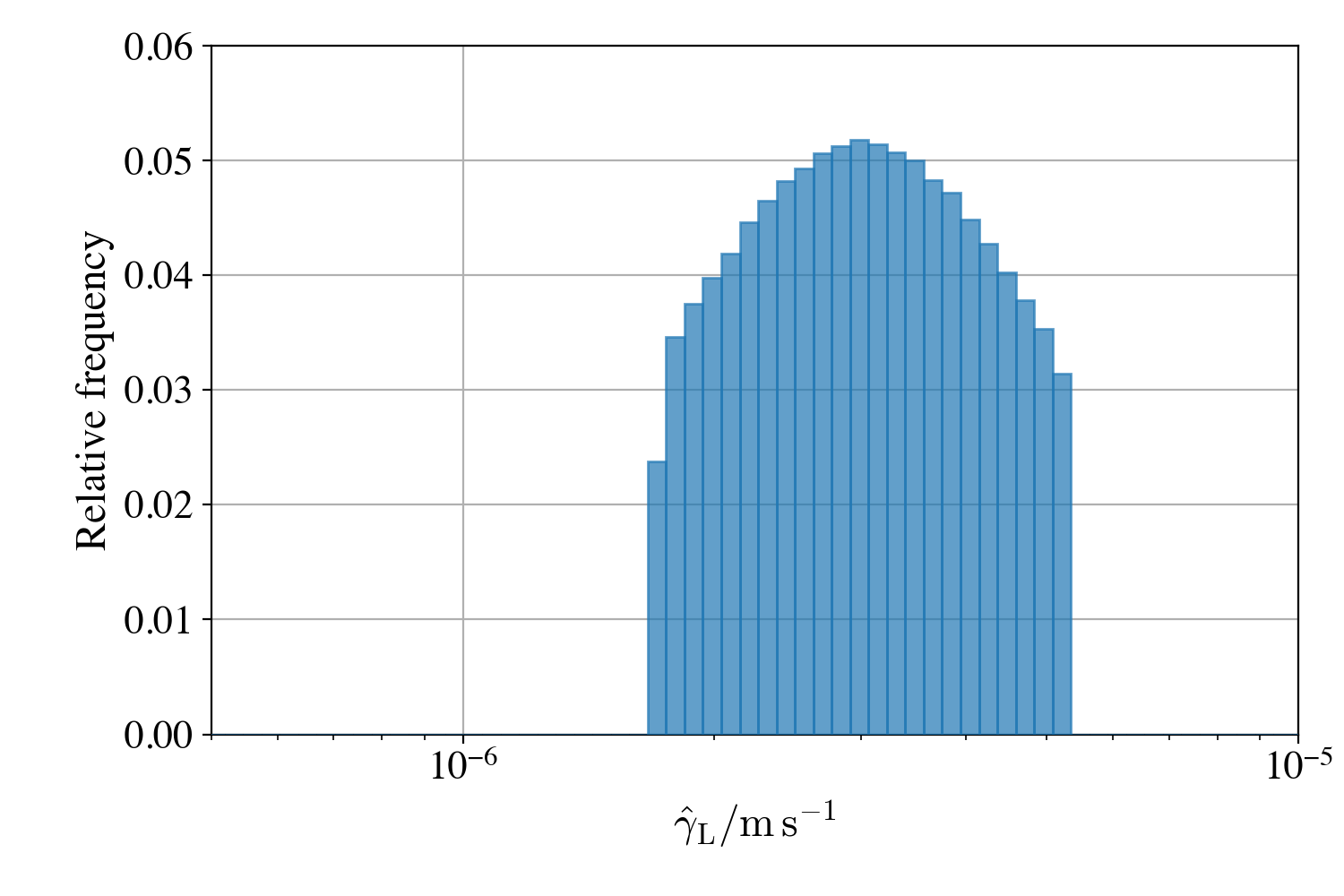}
    \caption{Histogram of the varied $\Hplus$ permeability values over the membrane using $l_{\mathrm{L}} = -5.77, \; u_{\mathrm{L}} = -5.27,  \;\mu_{\mathrm{L}} = -5.52, \; \sigma_{\mathrm{L}} = 0.25$.}
    \label{fig:gl_trunc_lognorm}
\end{minipage}
\hspace{0.1em}
\begin{minipage}[t]{.32\textwidth}
    \centering
    \includegraphics[width=0.99\textwidth]{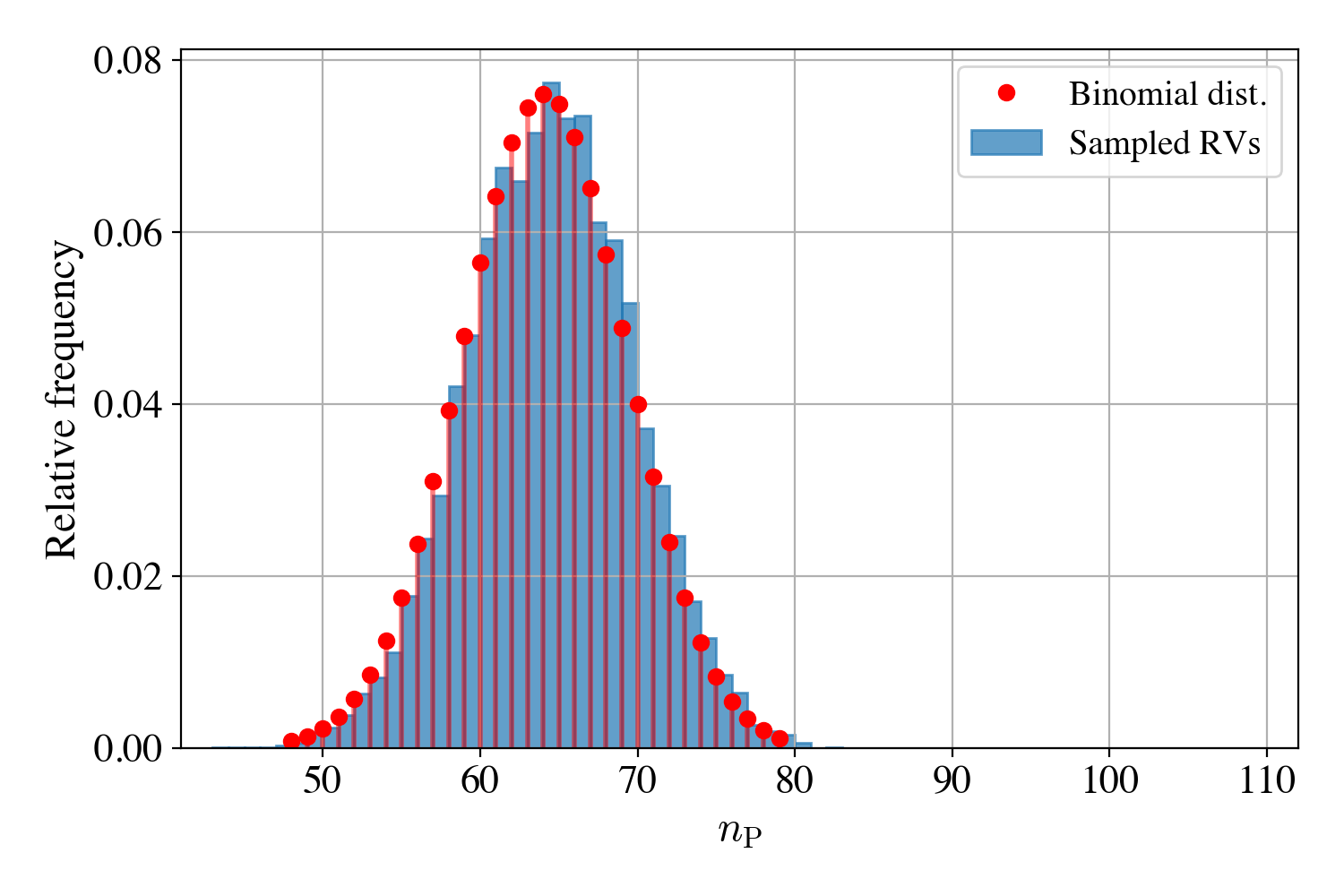}
    \caption{Binomial distribution (red) and histogram of generated Binomially distributed random variables (blue bars) for the number of pumps, $\np$, in the vesicle membrane of a vesicle with $\din = \SI{117.67}{\nano\meter}$, which corresponds to the mean of the log-normal distribution (in \cref{fig:din_exp_dist_fit}), and $\ntot = 112$ (for the given $\din$).}
    \label{fig:np_binom_fit}
\end{minipage}
\end{figure}

\begin{figure}
\centering
\includegraphics[width=0.8\textwidth]{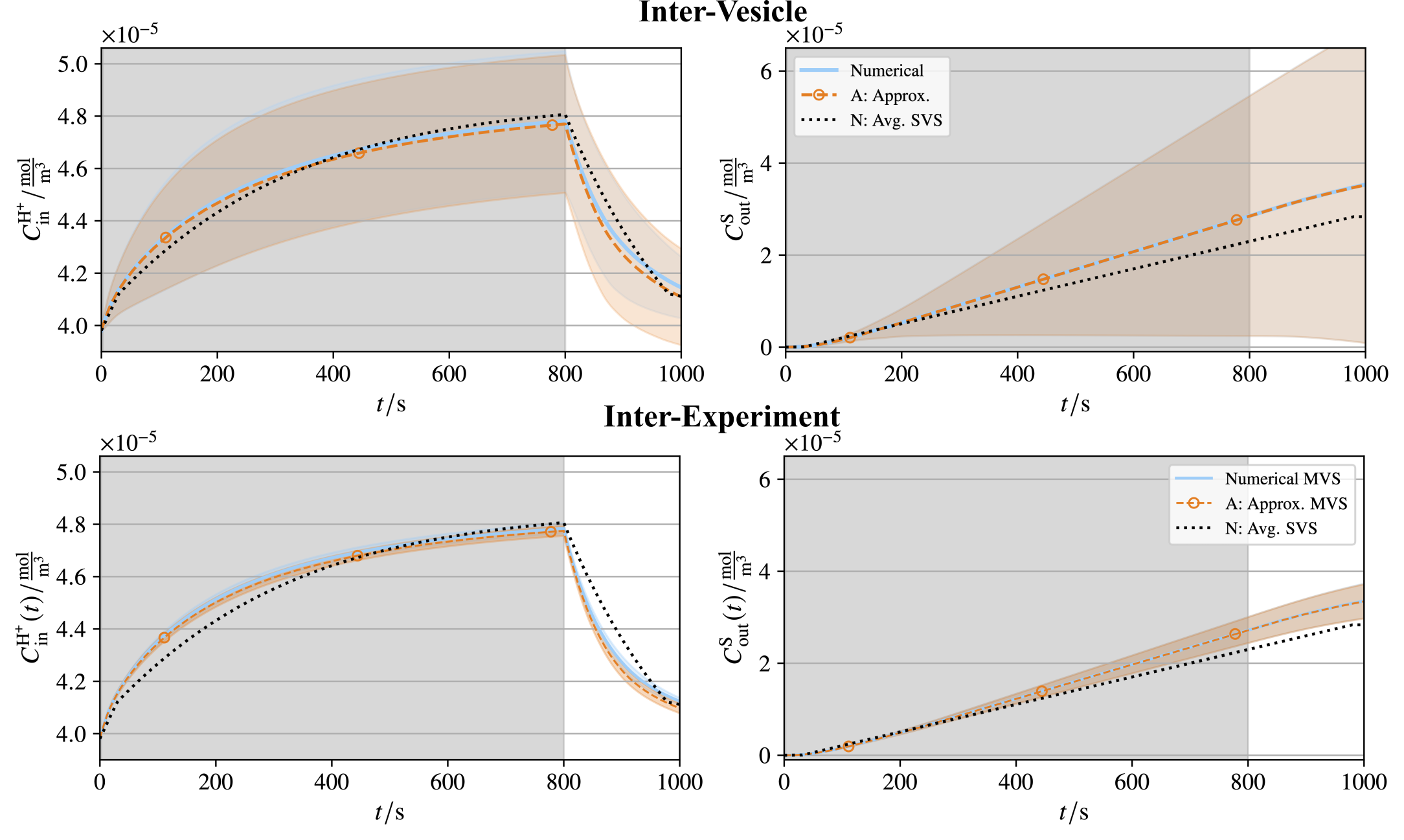}
\vspace{-2em}
\caption{Top: Inter-vesicle intravesicular $\Hplus$ and extravesicular $\Sub$ concentration mean (solid or dashed line) $\pm$ one standard deviation (shaded area). Bottom: Inter-experiment $\Hplus$ and extravesicular $\Sub$ concentration mean (solid line) $\pm$ one standard deviation (shaded area). Here, $\nMod = 100$ vesicles per simulation were modeled individually and $\nEx = 10$ simulations were averaged for the bottom row. The shaded area corresponds to the illumination time.}
\label{fig:interandintra}
\end{figure}

\subsection{Varying Mean Values for Parameters in Different Experimental Batches}

\noindent In practice, when multiple experiments are conducted using different batches of proteins, buffers, and polymers, the distributions of the parameters may vary between experiments. To account for such effects, we conduct a sensitivity analysis to investigate the influence of practically relevant parameters, which should be carefully tuned in order to reduce inter-experiment variance. Such analyses are beneficial for future experimental realizations of the system where the variance between batches has to be taken into account. For the analysis, we focus on the mean inner vesicle diameter $\dinMean$ and the mean leakage velocity over the membrane $\ghlHatMean$. These values may vary depending on the production process, which cannot be controlled perfectly. For the following results, $\nEx = 10$ experiments were conducted per setting, each consisting of $\nves = 100$ individually modeled vesicles.

\paragraph{Varying Mean Inner Vesicle Diameters}

\begin{figure}
    \centering
    \includegraphics[width=0.8\textwidth]{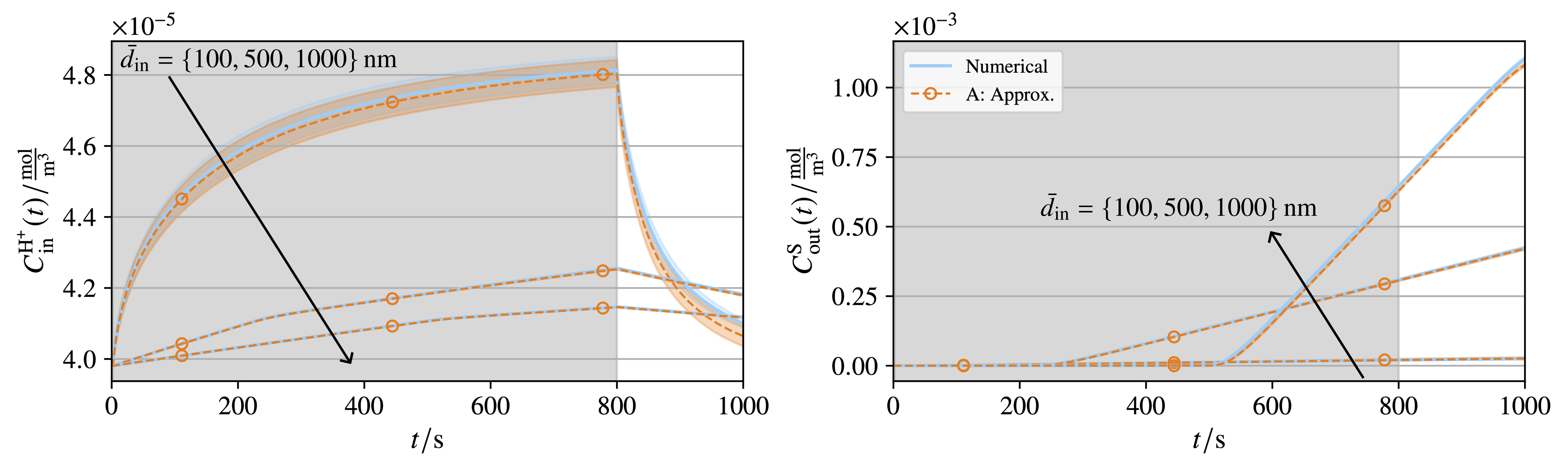}
    \vspace{-1.5em}
    \caption{Mean and one standard deviation of the intravesicular $\Hplus$ (left) and extravesicular $\Sub$ concentrations (right) for varying mean inner vesicle diameters. The chosen mean vesicle diameters were $\dinMean \in \{100, 500, 1000\} \; \si{\nano\meter}$. The shaded area corresponds to the illumination time.}
    \label{fig:var_din}
\end{figure}

\begin{figure}
    \centering
    \includegraphics[width=0.8\textwidth]{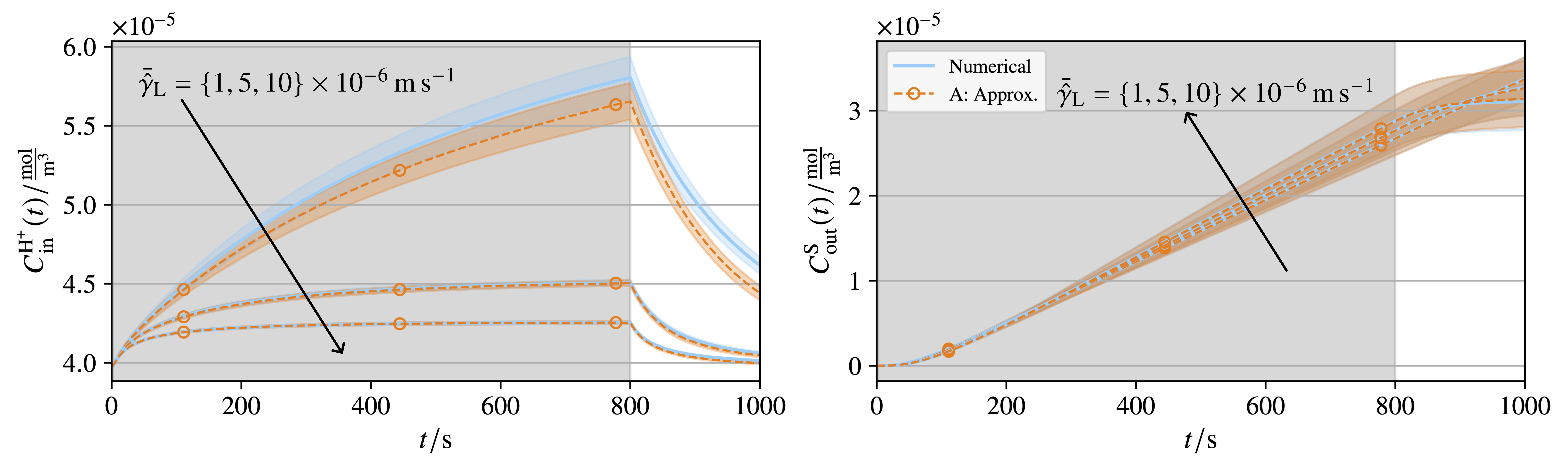}
    \vspace{-1.5em}
    \caption{Mean and one standard deviation of the intravesicular $\Hplus$ (left) and extravesicular $\Sub$ concentrations (right) for varying mean $\Hplus$ permeability coefficient, $\ghlHatMean$. The chosen permeability coefficients were $\ghlHatMean \in \{1, 5, 10\} \times 10^{-6} \; \si{\meter\per\s}$. The shaded area corresponds to the illumination time.}
    \label{fig:var_gl}
\end{figure}

\Cref{fig:var_din} shows the impact of different mean vesicle diameters $\dinMean$ on the inter-experiment variance for $\nEx = 10$ experiments, while the variance of its distribution remains unchanged. As $\dinMean$ directly influences $\Vin$, it also scales $\Chin{t}$ (see left panel in \cref{fig:var_din}). For example, larger vesicles exhibit a lower intravesicular concentration change upon illumination. Therefore, we also observe a delayed start of symport for these vesicles (see right panel in \cref{fig:var_din}). Similarly, as $\nsym$ increases with $\din$ and directly scales the symport rate, larger vesicles have a higher symport rate, which is reflected in a higher incline of $\Csout{t}$ during symport. Therefore, a quicker $\Sub$ release during symport by large vesicles comes at the trade-off of lower system responsivity. The leakage and symport rate are proportional to the outer vesicle area, which is linearly dependent on $\din^{2}$. However, as $\Vin$ grows linearly in $\din^{3}$, the effect of increasing $\Vin$ dominates the effect of increasing leakage and symport rate, as can be observed from the curves for $\Chin{t}$ in \cref{fig:var_din}. We observe that the absolute inter-experiment variance of $\Chin{t}$ is largest for small mean vesicle diameters in \cref{fig:var_din}. The cause of this is also the cubic dependence of $\Vin$ on $\din$. In conclusion, \cref{fig:var_din} indicates that changes in $\dinMean$ greatly impact the release of $\Sub$. Thus, parameter $\dinMean$ should be tuned carefully to obtain the desired system behavior.

\paragraph{Varying Mean Permeability Coefficient}

\noindent \Cref{fig:var_gl} shows a similar analysis for different mean $\Hplus$ permeability coefficients $\ghlHatMean$ of the membrane. Similar to \cref{fig:var_din}, the results were obtained using \ac{MVS} and averaging over $\nEx = 10$ simulations. For this analysis, the values of $\ghlHatMean$ were varied between $10^{-6} \si{\meter\per\s}$ and $10^{-5} \si{\meter\per\s}$. The left panel shows that there is a correlation between low permeability coefficients and high $\Chin{t}$ during illumination. Naturally, this also implies that our analytical solution performs worse for lower permeabilities as the buffering effect is modeled inadequately (refer to \cref{sec:mol}). Interestingly, for the chosen threshold $\thresh$, $\Csout{t}$ does not depend as strongly on $\ghlHatMean$ (see right panel in \cref{fig:var_gl}). However, the differences between the symport end times $\tesym{1}$ after $\tep{1} = \SI{800}{\s}$ can be seen clearly: For $\ghlHatMean = 10^{-5}\si{\meter\per\s}$, $\Csout{t}$ quickly stagnates, indicating that the symport ends shortly after the end of illumination. This is directly related to the lower $\Chin{t}$ for the respective permeability coefficient. A lower permeability coefficient and, thus, higher $\Chin{t}$ causes the symporters to remain active after illumination ends, which is why we cannot observe any stagnation in $\Csout{t}$ for $\ghlHatMean = 10^{-6} \si{\m \per \s}$. We observe that the changes in $\Csout{t}$ are less substantial when varying $\ghlHatMean$ than when varying $\dinMean$. Note that the y-axes in \cref{fig:var_din,fig:var_gl} differ by a magnitude of 100. This is because, unlike $\dinMean$, $\ghlHatMean$ scales only some transport rate (it does not scale the $\Sub$ transport rate of the release modules, for instance). Therefore, if the envisioned application is sensitive to the amount of released $\Sub$, it is important to keep $\dinMean$ constant among experimental realizations, while small deviations in $\ghlHatMean$ are likely less impactful.

\section{Conclusions}
\label{sec:conclusion}

\noindent In this paper, we introduced a new \ac{ND} design that can be used as an optically controllable \ac{TX} in synthetic \ac{MC} systems for the release of \acp{SM} via cooperating transmembrane proteins. The proposed modular design comprises an energizing module and a release module powered by the energizing module and has the potential to be useful in various future healthcare and industrial applications of \ac{MC} as it supports the controlled release of a variety of different \acp{SM}.
We proposed two analytical expressions for the concentrations of the involved molecules to describe the dynamics of the envisioned \ac{ND} and successfully verified them by comparison to a numerical baseline. Our model adequately captures real-world phenomena such as the presence of a pH buffer and \ac{SM} depletion in the vesicle. Our results demonstrate that the choice of appropriate system parameters, such as the ratio of pumps and co-transporters or the buffer molarity, is crucial for ensuring the feasibility of the desired system behavior. Consequently, the proposed analytical solutions can guide the design of future experiments and thereby accelerate the development time of the envisioned \ac{ND} by offering the possibility for \textit{in silico} optimization of the system parameters. 
Additionally, this work is a crucial step towards modeling experimentally feasible \ac{MC} \acp{TX} by accounting for the fact that practically implementable \ac{TX} systems consist of large groups of \acp{ND}. Moreover, we considered realistic random distributions for various system parameters caused by the chemical production of vesicular \acp{ND} and modeled \ac{ND} diversity.

\noindent An interesting topic for future work is a more detailed investigation of the behavior of the components of the system, e.g., by modeling the dependency that the applied light intensity has on the molecule transport rates and by developing more sophisticated models for the symport kinetics.

\appendix

\subsection{Proof of Theorem 1}
\label{sec:proof_mvs}

\noindent In order to prove \cref{theo:t1}, we need to show that one of the functions describing the $\Sub$ and $\I$ concentrations are not equal for the mean of the \ac{MVS} with varying parameters $\din$, $\np$, $\nsym$, and $\ghlHat$ and the corresponding \ac{SVS} with the mean parameters $\dinMean$, $\npMean$, $\nsymMean$, and $\ghlHatMean$. 
To this end, we use Jensen's inequality for convex function, $\phi: \mathbb{R} \rightarrow \mathbb{R}$ \cite[Sec. 2.7]{MacKay2002}:

\begin{equation}
    \phi(\mathcal{E}\{X\}) \le \mathcal{E}\{\phi(X)\},
    \label{eq:ji}
\end{equation}

\noindent where $X$ and $\mathcal{E}\{\cdot\}$ are a random variable and the expectation operation, respectively. The inequality in \eqref{eq:ji} becomes an equality iff $X$ is constant, i.e., deterministic, or iff the function $\phi$ is affine, i.e., linear \cite[Sec. 8.4]{Lin2011}. 
Let us consider \eqref{eq:csin_aci} and random variable $X := \nsym$. If we rewrite $\Csin{t}$ for the case of symporter activity in terms of $X$ and some constants $\beta_i$ with $i \in \{1, 2, 3, 4\}$, we obtain:

\begin{equation}
    \Csin{t, X} = \beta_1 \lamW{\beta_2 \mathrm{exp}(\beta_3 - \beta_4 X (t - \tSecStart{t}))}.
\end{equation}

\noindent Now, let us take the second derivative of $\Csin{t, X}$ with respect to the random variable $X$. We consequently obtain:

\begin{equation}
    C^{\mathrm{S} ''}_{\mathrm{in}}(t,X) = \beta_1(\beta_4(t - \tSecStart{t}))^2 \frac{\lamW{\beta_2 \mathrm{exp}(\beta_3 - \beta_4 X (t-\tSecStart{t}))}}{(1+ \lamW{\beta_2 \mathrm{exp}(\beta_3 - \beta_4 X (t-\tSecStart{t}))})^{3}}.
\end{equation}

\noindent This equation holds for all $t \ge \tSecStart{t}$. All $\beta_i$ are strictly positive. Additionally, the Lambert W--function yields positive values for positive arguments. Consequently, we observe that $ C^{\mathrm{S} ''}_{\mathrm{in}}(t,X) \ge 0$, i.e., $\Csin{t,X}$ is convex. Therefore, we can use Jensen's inequality to relate the function value at the expectation of $X$ to the expectation of the function evaluated at $X$. As $X$ is a (positive) non-deterministic random variable, we conclude that the following holds:

\begin{equation}
    \mathcal{E}\{\Csin{X}\} \neq \Csin{\mathcal{E}\{X\}}.
\end{equation}

\noindent This suffices to conclude that $\Csin{t}$ in an \ac{SVS} with average parameters does not correspond to the signal in the overall \ac{MVS}. Therefore, we have shown that the \ac{MVS} behavior cannot be captured by the models derived for the \ac{SVS} with average parameters if the symporters are activated at some point. Note that similar results can be obtained when considering the other random parameters, but the derivations are omitted here. \hfill $\square$

\vspace{-2em}
\bibliographystyle{IEEEtran}
\bibliography{literature}

\section*{Acknowledgment}

This work was funded by the Deutsche Forschungsgemeinschaft (DFG, German Research Foundation), GRK 2950, Project-ID 509922606.

\end{document}